\documentclass[journal,comsoc]{IEEEtran}

\normalsize

\usepackage[utf8]{inputenc} 
\usepackage[T1]{fontenc}
\usepackage{url,comment}
\usepackage{ifthen}
\usepackage{cite}
\usepackage{graphicx}
\usepackage{amsmath,amssymb,amsthm}
\usepackage{cite,soul}

\usepackage{multirow,rotating}
\usepackage{diagbox}
\usepackage{threeparttable}
\DeclareMathOperator{\parD}{Pareto}
\DeclareMathOperator{\erD}{Erlang}

\DeclareMathOperator{\expD}{Exp}
\DeclareMathOperator{\sexpD}{S-Exp}
\DeclareMathOperator{\BMD}{Bi-Modal}

\DeclareMathOperator{\expec}{\mathbb{E}}

\newtheorem{theorem}{Theorem}
\newtheorem{lemma}{Lemma}
\newtheorem{corollary}{Corollary}
\newtheorem{conjecture}{Conjecture}
\newtheorem{proposition}{Proposition}
\usepackage{booktabs}
\usepackage[colorlinks,pdfstartview=FitH,citecolor=blue,linkcolor=blue,urlcolor=blue]{hyperref}
\usepackage{tikz}
\usepackage[capitalize]{cleveref}

\usepackage[normalem]{ulem}
\usepackage{xcolor}
\makeatletter
\def\squiggly{\bgroup \markoverwith{\textcolor[rgb]{1,0,0}{\lower3.5\p@\hbox{\sixly \char58}}}\ULon}
\makeatother

\newcommand{\prob}[1]{{\mathbb P}}
\ifCLASSINFOpdf
\else
\fi

\begin{document}
%
\title{Diversity/Parallelism Trade-off\\ in Distributed Systems with Redundancy}
%
%
%
\author{Pei Peng, Emina Soljanin,~\IEEEmembership{Fellow,~IEEE,}
Philip Whiting
\thanks{Manuscript received October 5, 2020; revised May 9, 2021; accepted October 8, 2021. This work was supported  by  the  National  Science  Foundation  under  Grant  No.  CIF-1717314. An earlier version of this paper was presented in part at the Proceedings of 2020 IEEE International Symposium on Information Theory (ISIT)\cite{peng2020diversity} [DOI: 10.1109/ISIT44484.2020.9174030]. \emph{(Corresponding author: Pei Peng.)}}%
\thanks{Pei Peng and Emina Soljanin are with the Department
of Electrical and Computer Engineering, Rutgers, The State University of New Jersey, Piscataway, NJ 08854, USA (e-mail: pei.peng@rutgers.edu; emina.soljanin@rutgers.edu).}
\thanks{Philip Whiting is with School of Engineering, Macquarie University, Sydney, NSW 2109, Australia (e-mail: philip.whiting@mq.edu.au).}%
\thanks{Communicated by M. A. Maddah-Ali, Associate Editor for Communications.}
\thanks{Color versions of one or more figures in this article are available at \url{https://doi.org/10.1109/TIT.2021.3127920}.}
\thanks{Digital Object Identifier 10.1109/TIT.2021.3127920}
}%
%
%

\markboth{IEEE Transactions on Information Theory, to appear in  2021} {Submitted paper}
%



\maketitle

\begin{abstract}
Distributed computing enables {\it parallel} execution of smaller tasks that make up a large computing job. Its purpose is to reduce the job completion time. However, random fluctuations in task service times lead to straggling tasks with long execution times. Redundancy provides {\it diversity} that allows job completion when only a subset of redundant tasks is executed, thus removing the dependency on the straggling tasks. Under constrained resources (here, a fixed number of parallel servers), increasing redundancy reduces the available resources for parallelism. In this paper, we characterize the {\it diversity vs.\ parallelism} trade-off and identify the optimal strategy among replication, coding, and splitting, which minimizes the expected job completion time. We consider three common service time distributions and establish three models that describe the scaling of these distributions with the task size. We find that different distributions with different scaling models operate optimally at different redundancy levels, thus requiring very different code rates. 
\end{abstract}


%

\IEEEpeerreviewmaketitle

\section{Introduction}
Distributed parallel computing has become necessary for handling machine learning and other algorithms with ever increasing complexity and data requirements. This is because it provides simultaneous execution of smaller tasks that make up larger computing jobs. However, the large-scale sharing of computing resources causes random fluctuations in task service times\cite{dean2013tail}. Therefore, although executed in parallel, some tasks, known as stragglers, take much more time to complete, which consequently increases the job service time. Redundancy, in the form of simple task replication, and more recently, erasure coding, has emerged as a potentially powerful way to shorten the job execution time. Task redundancy allows job completion when only a subset of redundant tasks get executed, thus avoiding stragglers, see e.g.\ \cite{aktacs2019straggler,gardner2015reducing,ananthanarayanan2013effective,kadhe2015availability,kadhe2015analyzing,halbawi2018improving,ForkJoin:AktasKSS21,dutta2019optimal,li2015coded,yang2017coded,yuster2005fast,reisizadeh2019coded,tandon2017gradient,wang2015using,Codes&Qs:JoshiSW17,primorac2021hedge} and references therein.
Redundancy provides diversity since job completion can be accomplished in different ways, e.g., when any fixed-size subset of tasks gets executed.

In distributed, parallel computing with redundancy, both parallelism and diversity are essential in reducing job service time. However, both parallelism and diversity are provided by the same limited system's resources dedicated to the job execution, e.g., a fixed number of servers. To understand this tension on the system's resources that  parallelism and diversity bring about, let us consider two extreme ways to assign a job to $n$ servers. One is {\it splitting} or maximum {\it parallelism} with no redundancy. Here, the job is divided equally among the $n$ workers, and thus it gets completed when all workers execute their tasks.  The other is $n$-fold {\it replication} or maximum {\it diversity}. Here, the entire job is given to each worker, and thus it gets completed when at least one of the workers executes its task. Roughly speaking, splitting (maximal parallelism) is appropriate for large jobs with almost deterministic service time (i.e., no straggling servers), and replication (maximal diversity) is appropriate for small jobs with highly variable service time (i.e., many straggling servers).

 Given a fixed number of workers $n$, the general question is how much parallelism vs.\ diversity should be used. Consider a coding scheme where jobs are split in $k$ tasks and encoded into $n\ge k$ s.t.\ execution of any $k$ is sufficient for job completion. On the one hand, the smaller the $k$, the larger the task each server is given to execute. On the other hand, the smaller the $k$, the smaller the subset of tasks necessary for job completion. The choice of $k$ thus dictates the trade-off between parallelism (increasing with $k$)  and diversity (decreasing with $k$). We are here concerned with characterising the diversity vs.\ parallelism trade-off for different service time and task execution models, i.e., with finding an optimal $k$ for a given $n$.

There is a large body of literature on replication and erasure codes for classical machine learning and other algorithms (see e.g. \cite{CC:duttaCP16,li2015coded,lee2017speeding,dutta2017coded,yang2017coded,yuster2005fast,reisizadeh2019coded,yu2017polynomial,tandon2017gradient,yu2020straggler,merrill2016merge,karakus2019redundancy,kiani2018exploitation,ozfatura2019speeding,raviv2018erasure,ye2018communication,ferdinand2018anytime} and references therein), and thus it is reasonable to assume that codes exist for many job types and any $n$ and $k$ combination. However, very little is known about what exact $n$ and $k$ should be selected in a given scenario in order to optimize a particular metric or goal of interest. 
When the goal is only to have the job completed by a certain time and it is known that at most $\ell$ workers will not respond by that time, then simply setting $k=n-\ell$ will achieve the goal. 
However, service time is a random variable, and one can only talk about the probability of task completion by a certain time \cite{chen2010analysis,reiss2012towards,ananthanarayanan2014grass,wang2015using}. Therefore, the question one should ask is which $k$  minimizes the expected job completion time.

Several recent papers asked how much redundancy should be used in distributed systems. In particular, \cite{gardner2017better,Codes&Qs:JoshiSW17,aktas2019optimizing} are solely concerned with replication systems, and their results do not easily extend to erasure coding system. Coding systems were considered in e.g., \cite{joshi2012coding, joshi2014delay,aktacs2019straggler}. However, the system's resources were not assumed to be limited, and thus the diversity vs.\ parallelism trade-off was not addressed.

To study the diversity vs.\ parallelism trade-off in systems with limited resources, we need to know the service time probability density function (PDF) as well as how it {\it scales} (changes) with the size of the task. Various service time PDFs have been adopted in the literature. 
For theoretical analysis, $\parD$ distribution was used in  e.g.,\cite{aktacs2019straggler,MultServersHeavyTailedWorkload:PsounisMP05, UnfairnessSRPT:BansalH01, LowLatencyViaRed:VulimiriGM13, aktas2017effective, aktas2018straggler,joshi2018synergy}, 
$\erD$ was used in e.g.,\ \cite{Erlang:KwonG16,Erlang:VasantamMM17,Erlang:GorbunovaZM16,Erlang:PoloczekC16,Erlang:Thomasian15,Erlang:BanerjeeKA14}, 
Shifted-Exponential was used in e.g., \cite{joshi2012coding,liang2013fast, CC:duttaCP16,bitar2020minimizing},
the exponential distribution was used in \cite{gardner2015reducing, SimplexQs:AktasNS17, joshi2014delay}, and a certain type of Bi-Modal distribution was used in e.g.,\cite{behrouzi2019redundancy}, 
Some general classes of distributions (log-concave/convex) were considered in \cite{QC:JoshiSW15,Codes&Qs:JoshiSW17}.
The experiments with Amazon EC2 servers reported in \cite[Fig.~2]{GradientCoding:TandonLD17}
 show that the service time  can be modeled by a Bi-Modal distribution (e.g., the one we use in this paper) or a heavy tail distribution (e.g., Pareto). On the other hand, a recent system-level work justifies the exponential distribution by considering their experimental results running on AWS \cite{primorac2021hedge}.

There is no consensus on how the service time PDFs scale with the task size. For example, if the service time for some unit size task is Exponential, then some models assume that the service time for an $s$ times larger task is also Exponential with the $s$ times larger mean (i.e., scaled exponential \cite{CC:duttaCP16,joshi2014delay,reisizadeh2019coded}), while other models assume that it is Erlang (sum of $s$ exponential PDFs). If the service time for some unit size task is Shifted-Exponential, then some models assume that the service time for $s$ times larger task is also Shifted-Exponential, only with an $s$ times larger shift \cite{CodedStorage:HuangSX12,joshi2012coding, lee2017speeding}. More sophisticated models studied how job size changes the tail of the Pareto service time \cite{aktacs2019straggler}.
In this paper, we consider a number of common service time and scaling models. We find that different models operate optimally at very different levels of redundancy, and thus may require very different code rates. The contributions of the paper are stated in more detail in Sec.~\ref{sec:contributions}, after the computing system model is given in detail.  

The paper is organized as follows: In Sec.~\ref{sys:model}, we present the system architecture and the models for the service time and its scaling. In Sec.~\ref{sec:contributions}, we state the problem and summarize the contributions of this paper. In Sec.~\ref{sec:shift}, \ref{sec:Pareto}, and \ref{sec:Bi-Modal}, we  characterize the diversity vs.\ parallelism trade-off for three common service time distributions with three different scaling models. 
Conclusions are given in Sec.~\ref{sec:conclusions}.

\section{System Model and Problem Formulation}
\label{sys:model}
\subsection{System Architecture}
We adopt a system model as shown in Fig.~\ref{fig:MasterRole}, consisting of 
a single front-end master server and $n$ computing servers we refer to as workers. Such distributed, parallel computing system architectures where a single master node manages a computing cluster of nodes are commonly implemented in modern frameworks, e.g., Kubernetes~\cite{KubernetesDoc18} and Apache Mesos~\cite{Mesos:HindmanKZ11}.

\begin{figure}[hbt]
\begin{center}
\begin{tikzpicture}
\node at (2.2,0) {\includegraphics[width=0.47\textwidth]{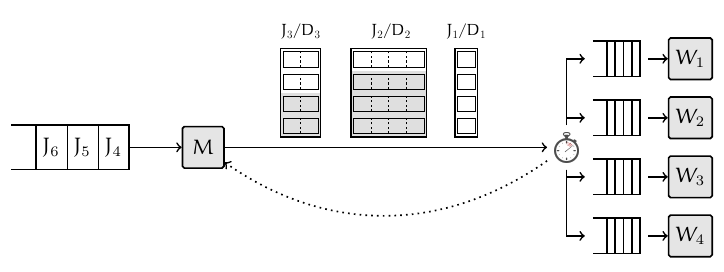}};
\end{tikzpicture}
\end{center}

\caption{\textbf{A Distributed Computing System:} Master node $M$ partitions jobs $J_i$ into tasks, possibly generates redundant tasks, and dispatches them to workers $W_1, W_2 , W_3, W_4$.
Shaded regions in the pre-processed jobs $J_2,J_3$ indicate redundancy. Here, each job consists of 4 computing units. Job $J_1$ is executed with maximal parallelism (splitting), $J_2$ with maximal diversity (replication), and $J_3$ is encoded by a $[4,2]$ erasure code.
\label{fig:MasterRole}}
\end{figure}

\subsection{Computing Jobs, Tasks, and Units}
\label{subsec:comp}
We are concerned with computing jobs that can be split into 
tasks which can then be executed independently in parallel by different workers. 
An example of such a job is vector by matrix multiplication, which is a basic operation in regression analysis and PageRank, and also in gradient descent, which resides at the core of almost any  machine learning algorithm \cite{ML:HoCC13,ML:DaiKW15,ML:AbadiAB16,ML:ChenMB16,ML:GemullaNH11,ML:DeanCM12,ML:Krizhevsky14,ML:ZhangCL15}.

We assume that there is some minimum size task for a given job below which distributed computing would be inefficient, and  refer to it as the  \textit{computing unit (CU)}. A task given to each worker can have one or more computing units. For example, if the job is to find a product $A\cdot X$ of a $3\times n$ matrix $A$ and an $n\times 1$ vector $X$, the computing unit could be a scalar product of a row of $A$ and $X$. Let the matrix $A$ be split into two submatrices: a $2\times n$ submatrix $A_1$ and a $1\times n$ submatrix $A_2$, as shown in Fig.~\ref{fig:CUs}.  
\begin{figure}[hbt]
    \centering
    \includegraphics{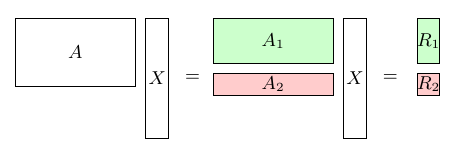}
    \caption{Multiplication $A\cdot X$ of a $3$-row matrix $A$ by vector $X$ is accomplished by parallel multiplication is of a $2$-row submatrix $A_1$ by vector $X$ and a $1$-row submatrix $A_2$ by vector $X$.}
    \label{fig:CUs}
\end{figure}
The job to compute $A\cdot X$, consisting of 3 CUs, can be split in two tasks for parallel execution: task $A_1\cdot X$ with two CUs, and task $A_2\cdot X$ with one CU.
We will measure the sizes of jobs and tasks by the number of their computing units. Although CU corresponds to the execution of identical tasks on different data sets, their execution times are not necessarily identical nor identically distributed, as we discuss in more detail below.

\subsection{Erasure Encoding Model}
\label{sec:coding}
We assume that each job consists of $n$ CUs where $n$ is the number of workers. The master node partitions each job into $k$ tasks, each of size $s=n/k$. It then generates $n-k$ redundant tasks, and dispatches the $n$ tasks to the $n$ workers. Therefore,  each worker is assigned a task of $s=n/k$ CUs. The redundant tasks are generated by an erasure code. If an $[n, k]$ code with minimum distance $d$ is used, then the job is completed when any $m = n - (d - 1)$ out of $n$ tasks are completed. The Singleton bound imposes the constraint $m \ge k$, where $m = k$ for MDS codes. In this paper, we limit our analysis to MDS codes for two main reasons: 1) MDS are the most common codes used in the literature on performance analysis of erasure coded distributed 
systems and 
2) MDS codes are sufficient to show the diversity/parallelism tradeoff and its dependence on multiple features in the system model, which is the purpose of this paper. However, our analysis is not limited to MDS codes, and we outline how it can be extended to general erasure codes in \Cref{sec:conclusions}.

 Fig.~\ref{fig:MasterRole} shows some possible ways in which the master server can preprocess a job, i.e., partition a job into CUs, group the CUs into tasks, and add redundancy. Because of redundancy, not all tasks assigned to the workers will be executed or even partially serviced. Because of that, we refer to the preprocessed jobs as virtual demands. Consider the virtual demands  $D_1$, $D_2$, and $D_3$ in the figure (resulting from processing jobs $J_1$, $J_2$, and $J_3$). Here, all jobs consist of 4 CUs. No redundant tasks are formed for job $J_1$, and thus the virtual demand $D_1$ and the original job are identical. Job $J_2$ is replicated on 4 workers, and thus the size of its virtual demand $D_2$ is 4 times the size of  $J_2$. Job $J_3$ is encoded by a systematic $[4,2]$ MDS code that generates 2 coded tasks of 2 CUs size. Its virtual demand $D_3$ is organized as follows: Workers $W_1$ and $W_2$ are each given a task consisting of 2 different CUs of $J_3$. Workers $W_3$ and $W_4$ are each given a coded task of 2 CUs size. 
 Job $J_1$ is handled by splitting, $J_2$ by replication, and $J_3$ by coding. Roughly speaking, the goal of this paper is to determine which of these three strategies should be used for several service time models for executing single and multiple CUs used in the literature. 
 
\subsection{Computing Unit Service Time Models \label{sec:servicetime}}
We model a computing unit service time as a random variable (RV) $\!X$, and refer to the tasks that are still running after a given time as stragglers. As we discussed in the introduction, there is no consensus on what the probability distribution of $X$ is.
We adopt the following three service time models commonly used in the literature.\\
\textbf{(Shifted)-Exponential:}\space $X \sim \sexpD(\Delta, W)$: Support of $X$ is $[\Delta, \infty)$, $\Delta$ is the minimum service time. The tail distribution is given as $\Pr\{X > x\} = e^{-(x-\Delta)/W}$ for $x > \Delta$.
The larger the $W$ the more likely straggling becomes. If $\Delta=0$, then $X \sim \expD(W)$ is exponential.\\
\textbf{Pareto:}\space $X \sim \parD(\lambda, \alpha)$: Support of $X$ is $[\lambda, \infty)$ where $\lambda$ is the minimum service completion time. Tail distribution is given as $\Pr\{X > x\} = (\frac{\lambda}{x})^{\alpha}$ for $x > \lambda$ where $\alpha$ is known as the tail index and models tail heaviness. Smaller $\alpha$ means a heavier tail, and thus more likely straggling.\\
%
\textbf{Simple Bi-Modal:}\space $X\sim \BMD(B, \epsilon)$: Under this distribution, $X$ takes only two values:
\begin{equation}
X = \begin{cases}
1  & \text{w.p.} \quad  1-\epsilon \\
B>1 & \text{w.p.} \quad \epsilon  ~\leftarrow \text{probability of straggling}
\end{cases} \label{eqn:toy_X_dist}
\end{equation}
 This distribution features two important aspects of service straggling: probability of straggling $\epsilon$ and magnitude of straggling $B$. 
\subsection{Service Time Scaling with the Task Size\label{sec:servicescaling}}
In the previous section, we listed three common models for the service time of a single computing unit. Together with those models, we will adopt three models for the service time of consecutive computing units that have frequently been used in the literature, as discussed in the introduction. The number of computing units that get assigned to each worker (that is, their task sizes) depends on the code rate $k/n$ used in the system, and thus these scaling models are relevant because they tell us how the task service time scales with its size.

  We consider three different commonly adopted models for the service time of consecutive CUs execution on the same server. For all three models, we assume independence across the servers. The models are described next, and their impact on the diversity vs.\ parallelism trade off is one of the main concerns of this paper.\\
\textbf{Model 1} -- {\it Server-Dependent Scaling:} \space
 The assumption here is that the straggling effect depends on the server and is identical for each CU executed on that server. Namely, there is some initial handshake time $\Delta$ after which the server completes its first and each subsequent CU in time $X$, i.e. $Y = \Delta + s\cdot X$. E.g. $X\sim\expD(W)$, then $Y\sim \sexpD(\Delta,sW)$. Note that $\Delta$ may be equal to $0$, giving $Y =  s\cdot X$. For example, when $X\sim\parD(\lambda,\alpha)$, then $Y\sim \parD(s\lambda,\alpha)$.
 \\
\textbf{Model 2} -- {\it Data-Dependent Scaling:}  \space The assumptions here are that 1) each CU in a task of $s$ CUs takes $\Delta$ time units to complete and  2) there are some inherent additive system randomness at each server which does not depend on the task size $s$ that determines the straggling effect  $X$. Therefore, $Y = s\cdot\Delta + X$. 
 \\
 \textbf{Model 3} -- {\it Additive Scaling:} \space
 The assumption here is that the execution times of CUs are i.i.d. Therefore, $Y=\sum_{i=1}^{s}X_{i}$ where $X_{1},X_{2},\cdots,X_{s}$ are independent.

\subsection{Job Completion Time}
As discussed above, the task execution times of the workers are i.i.d. RVs. The PDF of the RV $Y$ modeling task execution time depends on the assumed model for the execution of a single and multiple CUs. When an $[n,k]$ code is used, the job is complete when any $k$ out of $n$ workers complete their size $s$ tasks ($s=n/k$).
Thus, the job completion time is also an RV, which we denote by $Y_{k:n}$ since it represents the $k$-th order statistic of $n$ RVs distributed as $Y$.
Let $Y_1,Y_2, \dots, Y_n$ be $n$ samples of some RV $Y$. Then the $k$-th smallest is an RV, commonly denoted by $Y_{k:n}$, and known as the $k$-th order statistic of $Y_1,\dots Y_n$. 

\begin{center}\begin{small}
    \begin{tabular}{ rl}
       $n$  --  &  number of workers and also the job size in CUs\\
       $k$   --  & number of workers that have to execute their tasks\\& for job completion ( {\it diversity/parallelism parameter})\\
       $s$  -- & number of CUs per task, $s=n/k$\\
       $Y_{k:n}$ -- & job completion time when each worker's task size is $s$\\
    \end{tabular}\end{small}
\end{center}

\section{Problem Statement and Summary of the Contributions
\label{sec:contributions}}
Our goal is to characterise the expected job completion time $\expec[Y_{k:n}]$ for the service time and scaling models defined above. We are in particular interested in finding which $k$ (i.e., code rate $k/n$) minimizes $\expec[Y_{k:n}]$. Recall that when $k=1$, we have replication (maximum diversity, no parallelism), and when $k=n$, we have splitting (maximum parallelism, no diversity). When $1<k<n$, we use MDS coding and have a diversity/parallelism trade-off determined by the value of $k$.

The following table summarizes our findings by indicating whether splitting, replication, or coding minimizes the average job completion time $\expec[Y_{k:n}]$. Much more detail is given in the following sections.

\begin{table*}[hbt]
\begin{center}
\begin{threeparttable}
    \caption{ Strategies that minimize the average job completion time 
for a given service time PDF and scaling model}
    \label{Table:models}
		\begin{small}
			\begin{tabular}{@{}lllll@{}}
				\toprule 
				& &   \multicolumn{3}{c@{}}{\sc {service time PDF}}
				\\  \cmidrule{3-5}
				& & {\bf Shifted Exponential}
				& {\bf Pareto}
				& {\bf Bi-Modal}
				\\ [2mm]
				\multirow{5}{*}{\rotatebox{90}{
						{\sc{Scaling}}}} & 
				{\bf Server-Dependent} &  R \tnote{1} & S $\longrightarrow$ C \tnote{2}
				 &  S $\longrightarrow$ C $\longrightarrow$ S\\[2mm] 
				& {\bf Data-Dependent}   & S $\longrightarrow$ C $\longrightarrow$ R & S $\longrightarrow$ C $\longrightarrow$ R & S $\longrightarrow$ C $\longrightarrow$ S\\[2mm]
				& {\bf Additive}   & S $\longrightarrow$ C & S $\longrightarrow$ C & S $\longrightarrow$ C $\longrightarrow$ S\\[2mm]
				\bottomrule
			\end{tabular}
		\end{small}
\begin{tablenotes}
\footnotesize
\item[1] Strategies: R - replication, S - splitting, C - coding.
\item[2] $\longrightarrow$ indicates how the optimal strategy changes as the tail of the PDF becomes heavier (straggling becomes more likely).
\end{tablenotes}
     \end{threeparttable}
	\end{center}
\end{table*}

We consider the service time PDF and scaling models that are most commonly adopted in the literature. However, some of our results (we believe) can be extended to general service time PDFs, as we indicate by the claims and conjectures stated throughout the paper. These observations are relevant to practitioners who may have limited knowledge about their systems' behaviour.

To derive our results, we have relied on the following classical probabilistic models and arguments, which to the best of our knowledge, have not been previously used in this context. We introduce 
a generalized birthday problem to analyze the splitting strategy for the (Shifted-)Exponential service time with additive scaling. We recognize the stochastic dominance
of splitting over coding.
We show that the law of large numbers (LLN) can be used as an effective tool in finding the optimal code rate for systems with Bi-Modal service times and large number of workers. Moreover, we demonstrate how an LLN based analysis can be used to establish that for additive scaling and any service time distribution with the $4$-th moment, splitting is a  better strategy than replication for a sufficiently large number of workers.

\section{(Shifted-)Exponential Service Time\label{sec:shift}}
Under the (Shifted-)Exponential model, the CU service time is given by $\Delta+X$, where $X\sim \expD(W)$. The expected job completion time $\expec[Y_{k:n}]$ depends on the service time of $s=n/k$ CUs, which is determined by the service time scaling model. 
In the following three subsections, we determine $\expec[Y_{k:n}]$ for our three scaling models. Some results in this section were published in \cite{peng2020diversity}.

\subsection{Server-Dependent Scaling}
\label{sec:server-shift}
Under the server-dependent scaling, the service time of a task consisting of $s$ CUs is given by $Y=\Delta+s\cdot X$, which means $Y\sim \sexpD(\Delta, sW)$. Therefore, the job completion time is given by
\[
Y_{k:n}= \Delta + s\cdot X_{k:n}, ~~\text{where}~~ X\sim \expD(W)
\]
and, by using the expression for $\expec[X_{k:n}]$ in \eqref{eq.orderstat_mean_variance}, we have
\begin{equation}
\label{Eq:server-shift}
    \expec[Y_{k:n}]= \Delta+sW(H_{n}-H_{n-k}), ~~\text{where}~~ s=\frac{n}{k}
\end{equation}
The minimum expected job completion time is given by the following theorem: 
\begin{theorem} 
\label{thm:server}
The expected job completion time for $\sexpD(\Delta,W)$ service time with server-dependent scaling is minimized by replication (maximal diversity), i.e. $k=1$.
\end{theorem}
\begin{proof}
 From \eqref{Eq:server-shift}, we see that $\expec[Y_{k:n}]$ is an increasing function of $k$ for a given $n$, as follows
\begin{align*}
    \expec[Y_{k+1:n}] & = \Delta+W\frac{n}{k+1}(H_{n}-H_{n-k-1})\\
    & = \Delta+W\frac{n}{k+1}\Bigl(H_{n}-H_{n-k}+\frac{1}{n-k}\Bigr)\\
    & = \expec[Y_{k:n}]+\frac{Wn}{k+1}
    \Bigl[\frac{1}{n-k} - \frac{1}{k}(H_{n}-H_{n-k})\Bigr]
\end{align*}
Since the term in square brackets above is positive, we have $\expec[Y_{k+1:n}]>\expec[Y_{k:n}]$ for any positive integer $k\le n$.
\end{proof}
\paragraph*{Numerical Analysis}
We evaluate \eqref{Eq:server-shift} to see
how the expected job completion time $\expec[Y_{k:n}]$ changes with the diversity/parallelism parameter $k$. We consider a system with $n=12$ workers and the following six different combinations of $W$ and $\Delta$: $\Delta=1$ with $W\in\{0, 5, 10\}$, and  $W=1$ with $\Delta\in\{0, 5, 10\}$.  
\begin{figure}
    \centering
    \includegraphics[width=0.39\textwidth]{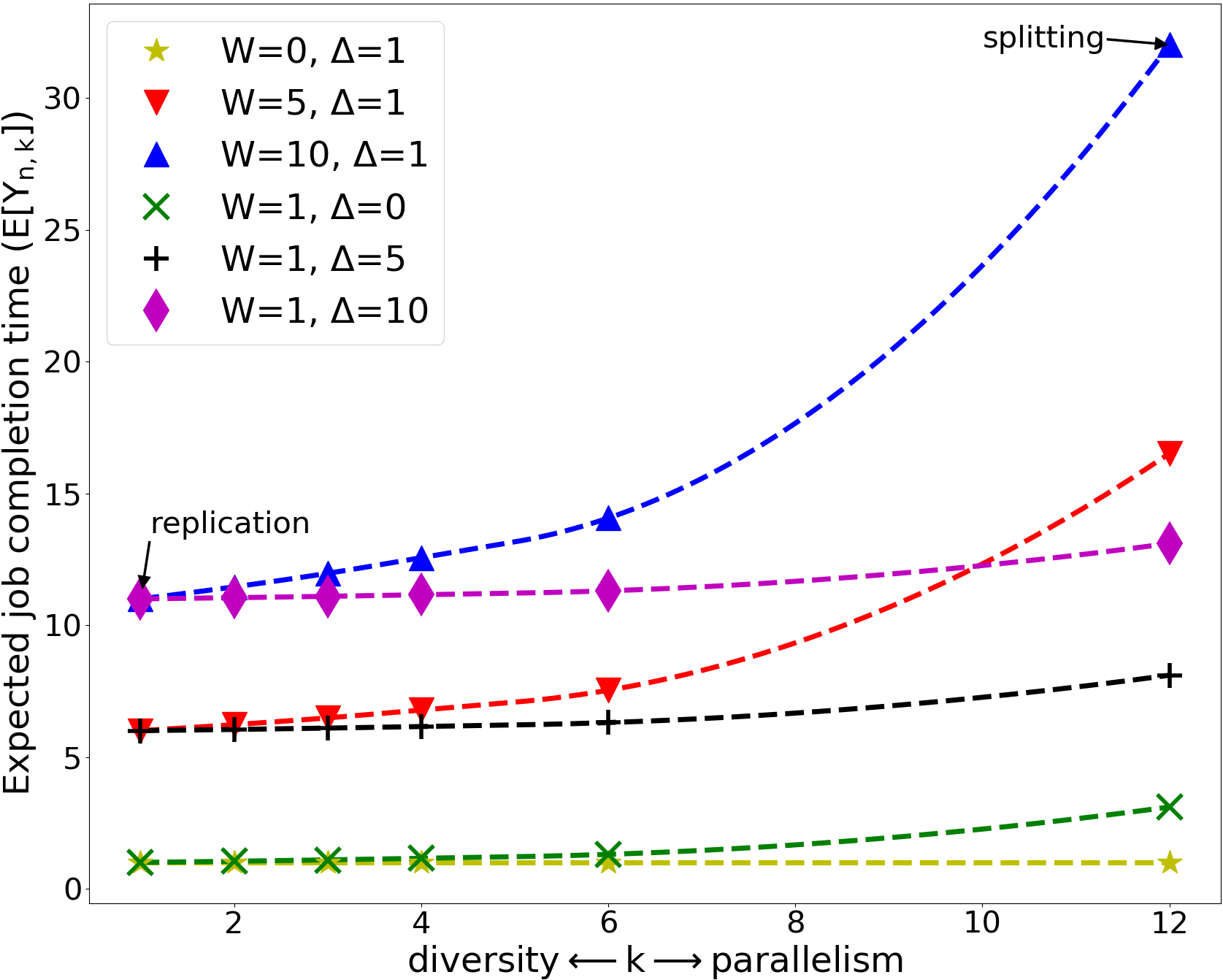}
    \caption{Expected job completion time $\expec[Y_{k:n}]$ for (Shifted-)Exponential service time with server-dependent scaling as a function of the diversity/parallelism parameter $k$ (cf.~\eqref{Eq:server-shift}). The number of workers (job size) is $n=12$, and task size per worker is $s=n/k$ (Since both $k$ and $s$ are integers, we have $k\in\{1,2,3,4,6,12\}$. We use  dashed curves to connect the points corresponding to different allowed values of $k$ for a given $\Delta$ and $W$ combination.)
    Replication is optimal for minimizing $\expec[Y_{k:n}]$.}
    \label{fig:scale}
\end{figure}
The results are plotted in Fig.~\ref{fig:scale}. $W=0$ corresponds to the  special scenario where $\expec[Y_{k:n}]=\Delta$, that is, the service time is deterministic and does not change with $k$. When $W>0$, $\expec[Y_{k:n}]$ reaches its minimum at $k=1$. When $W=1$ and $\Delta\in\{0, 5, 10\}$, $\expec[Y_{k:n}]$ increases with $\Delta$, but changes little with $k$. When $\Delta=1$ and $W\in\{0, 5, 10\}$, the slope of the corresponding curves increases with $W$. Although maximal diversity is optimal for all values of the parameters, it is much more effective in reducing the expected job completion time when $W$ is large compared to when $W$ is small. 

\subsection{Data-Dependent Scaling}
\label{sec:data}
Under the data-dependent scaling, the service time of a task consisting of $s$ CUs is given by $Y=s\cdot\Delta+X$, which means $Y\sim \sexpD(s\Delta, W)$, Therefore, the job completion time is given by
\[
Y_{k:n}= s\cdot\Delta + X_{k:n}, ~~\text{where}~~ X\sim \expD(W)
\]
and, by using the expression for $\expec[X_{k:n}]$ in  \eqref{eq.orderstat_mean_variance}, we have
\begin{equation}
\label{Eq:data-shift}
\expec[Y_{k:n}]=  s\Delta+W(H_{n}-H_{n-k})=W
\bigl[\frac{n}{k}\cdot \frac{\Delta}{W} + (H_{n}-H_{n-k})\bigr].
\end{equation}

\begin{theorem} 
\label{thm:data}
The expected job completion time for
$\sexpD(\Delta,W)$ 
service time with data-dependent scaling is minimal when $k=k^*$, 
where $ k^* =\mathop{\arg\min}\limits_{k}{W
\Bigl[\frac{nd}{k} + (H_{n}-H_{n-k})\Bigr]}$, $d = \Delta/W$. Furthermore, $k^*$ takes the value $\lceil n (-d/2 +\sqrt{d + d^2/4})\rceil$ or $\lfloor n (-d/2 +\sqrt{d + d^2/4})\rfloor$. 
\end{theorem}
\begin{proof}
The result is obtained by simple calculus using the $\log$ approximation to the harmonic numbers in \eqref{Eq:data-shift}.
\end{proof}
Note that this expression depends only on the ratio 
$d = \Delta/W$. For  $\Delta\gg W$ (large $d$), the service time is essentially deterministic and it is optimal to use maximum parallelism, that is, splitting ($k=n$) is optimal.
On the other hand, when $W\gg \Delta$ (small $d$) execution time is much more variable and it is optimal to operate with maximum diversity, that is, replication ($k=1$) is optimal (cf.\ \cite{joshi2012coding}).
\\[1ex]
{\it Numerical Analysis:}
We evaluate \eqref{Eq:data-shift} for $\expec[Y_{k:n}]$ vs.\ $k$. We consider a system with $n=12$ workers the following five different values of $W/\Delta$: 1. $W=0$ ($\Delta=10$); 2. $W/\Delta=0.1$ ($W=1$, $\Delta=10$); 3. $W/\Delta=1$ ($W=5$, $\Delta=5$); 4. $W/\Delta=10$ ($W=10$, $\Delta=1$); 5. $\Delta=0$ ($W=10$).
\begin{figure}
    \centering
    \includegraphics[width=0.4\textwidth]{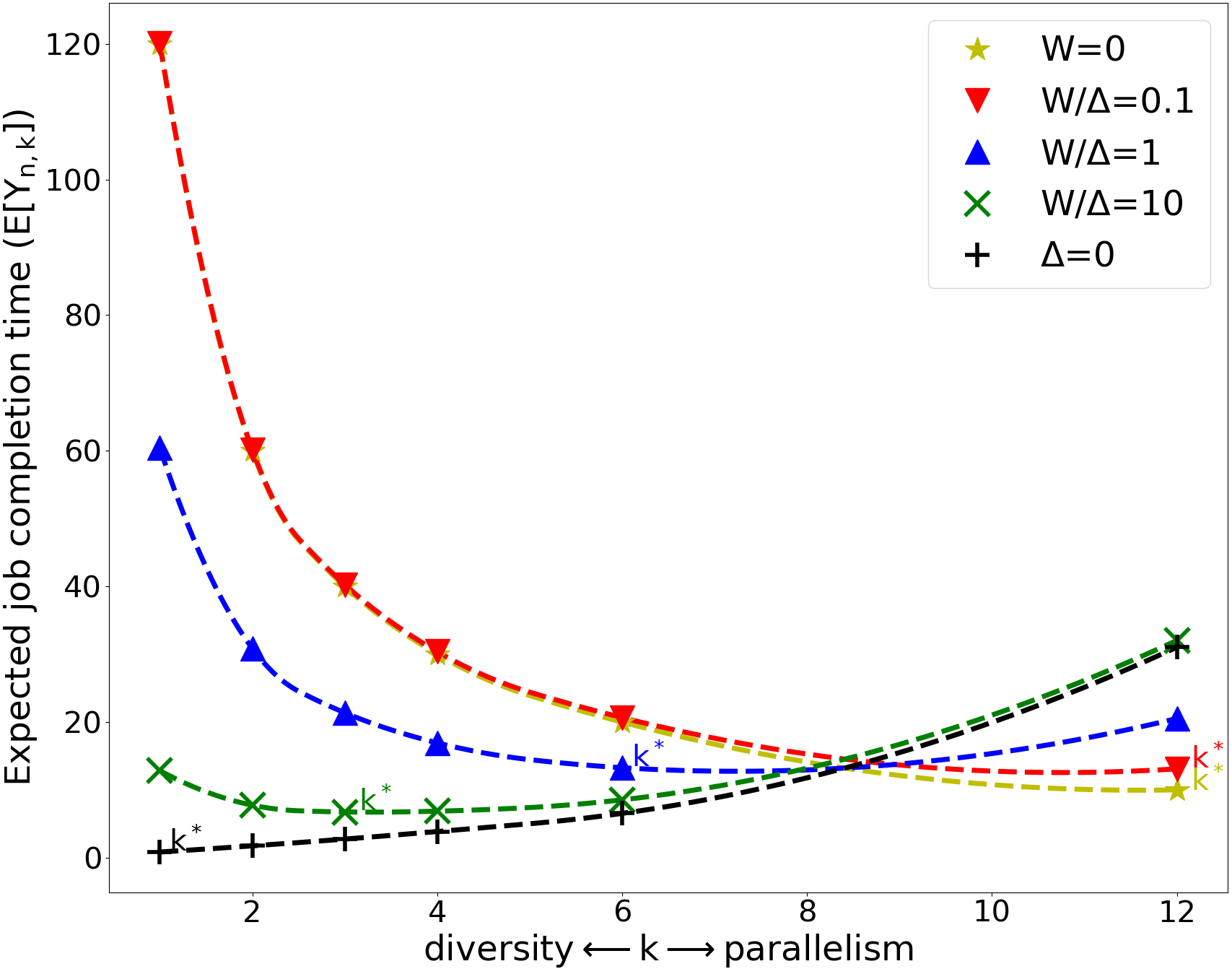}
    \caption{Expected job completion time $\expec[Y_{k:n}]$ for (Shifted-)Exponential service time with data-dependent scaling as a function of the diversity/parallelism parameter $k$ (cf.~\eqref{Eq:data-shift}). The number of workers (job size) is $n=12$, and task size per worker is $s=n/k$ (Since both $k$ and $s$ are integers, we have $k\in\{1,2,3,4,6,12\}$. We use dashed curves to connect the points corresponding to different allowed values of $k$ for a given $\Delta$ and $W$ combination.) Parallelism outperforms diversity for small $W/\Delta$ and vice versa if $W/\Delta$ is large. }
    \label{fig:shift}
\end{figure}
The results are plotted in Fig.~\ref{fig:shift}. By comparing different $W/\Delta$ scenarios, we conclude that when $W/\Delta$ (e.g. $0$, $0.1$) is small, then $\expec[Y_{k:n}]$ decreases as $k$, increases, which means that splitting is optimal. When $W/\Delta$ is large (and $\Delta=0$), the $\expec[Y_{k:n}]$ increases with $k$, which means that replication is optimal. Otherwise, $\expec[Y_{k:n}]$ reaches its minimum at $1 < k < 12$, which means that coding at a certain non-trivial rate is optimal. These observations are consistent with the theoretical analysis for $k^*$.

\subsection{Additive Scaling}
\label{sec:additive}
Under the additive scaling, the service time of a task consisting of $s$ CUs is given by $Y=s\cdot\Delta+(X_1+\dots +X_{s})=s\cdot\Delta+Z$, where $Z\sim \erD(s,W)$. Therefore, the job completion time is given by
\[
Y_{k:n}=s\cdot\Delta + Z_{k:n}, ~~\text{where}~~ Z\sim \erD(s,W)
\]
The expectation of the $k$-th order statistic of Erlang distribution is given by \eqref{Eq:Erlang-order}, 
which can be used for numerical results but is unsuitable for theoretical analysis. Asymptotics are available for large $n$ and $k= \mathcal{O}(1)$.
We now derive analytical expressions for the expected job completion time under splitting and replication, and show that splitting outperforms replication for sufficiently large $n$. We then show that rate $1/2$ coding outperforms splitting when $\Delta=0$. 

\paragraph{Splitting vs.\ Replication}
Under splitting, the job completion time is given by
\begin{equation}
Y_{n:n}=\Delta+X_{1:n}+X_{1:(n-1)}+\dots + X_{1:2}+X_{1:1}
\label{eq:AddSum}
\end{equation}
where $X_i$'s are i.i.d. $\expD(W)$, then $X_{1:n}$ is $\expD(W/n)$, and therefore,
\[
\expec[Y_{n:n}]= \Delta + W H_n
\]

Under replication, we have 
\[
\expec[Y_{1:n}]= n\Delta + W \frac{1}{n}\int_0^\infty \!\! e^{-t} \Bigl[R_n\Bigl(\frac{t}{n}\Bigr)\Bigr]^ndt 
\]
\[
~~ \text{where} ~~
R_n(x)=1+\frac{x}{1!}+\frac{x^2}{2!}+\dots+\frac{x^{n-1}}{(n-1)!}
\]
This result is a corollary of Theorem~\ref{Thm:birth}.

\begin{theorem}
\label{Thm:birth}
Let the service times of CUs be independent and exponential with rate 1.
 If a job with $d$ CUs is replicated over $n$ workers, then the expected job completion time is 
\begin{equation}
\label{Eq:add_repli}
    \frac{1}{n}\int_0^\infty \!\! e^{-t} \Bigl[S_d\Bigl(\frac{t}{n}\Bigr)\Bigr]^ndt 
\end{equation}
\[
~~ \text{where} ~~
S_d(x)=1+\frac{x}{1!}+\frac{x^2}{2!}+\dots+\frac{x^{d-1}}{(d-1)!}
\]
\end{theorem}
\begin{proof}
Let $t_1, t_2,\dots$ be time epochs at which a CU gets completed on any of the $n$ servers. Because all $d$ CUs of the job are replicated on each of the $n$ servers, the job is completed when $d$ CUs get completed on any single server, which happens at some time $t_{\ell_{d,n}}$. Note that $\ell_{d,n}$ is a random variable. We represent $t_{\ell_{d,n}}$ as a sum of the CU inter-completion times. 
\begin{equation}
t_{\ell_{d,n}}=\sum_{j=1}^{\ell_{d,n}} (t_j-t_{j-1}), ~~\text{where $t_0$ is set to 0.}
\label{eq:time}
\end{equation}
Note that 1) $t_j-t_{j-1}$ are independent and exponentially distributed with rate $n$ (the minimum of $n$ independent exponentials with rate 1), and 2) $t_j-t_{j-1}$ are independent from $\ell_{d,n}$. Observe next that Wald's identity (Ch.10.2 in \cite{grimmett2020probability}) can be applied to \eqref{eq:time}. Therefore,
\[
\mathbb{E}[t_{\ell_{d,n}}] =\frac{1}{n}\cdot \mathbb{E}[\ell_{d,n}] 
\]
Now observe that $\ell_{d,n}$ corresponds to a generalized birthday problem that the expected number of draws from $n$ coupons until a coupon shows up $d$ times. The claim follows from the result for $\mathbb{E}[\ell_{d,n}]$ in Appendix~\ref{Ap:birthday}.
\end{proof}

In Appendix~\ref{Ap:birthday}, we also have an asymptotic result for \eqref{Eq:add_repli}. For $n$ large, we further simplify the expression of $\expec[Y_{1:n}]$:
\begin{equation}
\label{Eq:approx}
    \expec[Y_{1:n}] \sim n\Delta + \frac{W}{n}\sqrt[n]{n!}\,\Gamma(1+1/n)n^{1-\frac{1}{n}}, ~~\text{as}~~n\rightarrow \infty
\end{equation}

\begin{theorem}
\label{th:splitvsrepli}
For large enough $n$, splitting (maximal parallelism) outperforms replication (maximal diversity).
\end{theorem}
\begin{proof}
By using the Stirling's formula for $\sqrt[n]{n!}$ in 
(\ref{Eq:approx}), we have
\begin{align*}
   &n\Delta + \frac{W}{n}\sqrt[n]{n!}\,\Gamma(1+1/n)n^{1-\frac{1}{n}}\\
   &\ge  n\Delta + \frac{W}{n}\sqrt[n]{\sqrt{2\pi}n^{n+\frac{1}{2}}e^{-n}}\,\Gamma(1+1/n)n^{1-\frac{1}{n}}\\
 &=n\Delta + \frac{W}{n}\sqrt[n]{\sqrt{2\pi}n^{n+\frac{1}{2}}e^{-n}}\,\frac{\Gamma(2+1/n)}{1+1/n}n^{1-\frac{1}{n}}\\
    & > n\Delta +
 \frac{W}{en}\sqrt[2n]{2\pi}\,n^{1+\frac{1}{2n}}n^{1-\frac{1}{n}}\frac{1}{1+1/n}
    >  n\Delta + \frac{W}{2e}n^{1-\frac{1}{2n}} 
\end{align*}
For a large enough $n$, $\expec[Y_{1:n}]$ is well approximated by $ n\Delta + \frac{W}{n}\sqrt[n]{n!}\,\Gamma(1+1/n)n^{1-\frac{1}{n}}$. Recall that $H_n=\mathcal{O}(\log n)$ and  $n^{1-\frac{1}{2n}}/2e>\sqrt{n}/2e=\Omega(\sqrt{n})$. Therefore,
\[
\expec[Y_{1:n}] > \Delta + WH_n=\expec[Y_{n:n}], ~~\text{as}~~n\rightarrow \infty
\]
Note that the theorem holds for $\Delta=0$.
\end{proof} 

\paragraph{Rate $1/2$ Coding, $s=2$}

We consider the special case when $\Delta=0$, $n$ is even, and $s=n/k=2$.
Therefore,
$k=n/2$ workers have to complete their two CUs in order for the job itself to be complete. 

Let $Y_{n:n}$ be the time to complete the job under splitting, as given by \eqref{eq:AddSum} for $\Delta=0$, and $Y_{n/2:n}$ the random time to complete the job under coding with $s=2$. Theorem~\ref{thm:sca}
below shows that ${\sf P}\{Y_{n/2:n}>x\} \le {\sf P}\{Y_{n:n}>x\}$. It follows that
\[
\expec[Y_{n/2:n}]\le \expec[Y_{n:n}]
\]
since for any non-negative random variable $X$, we have
$\expec[X]=\int_0^\infty{\sf P}(X>x)\,dx$.

It is, therefore, better to use a rate half code than splitting.

\begin{theorem}
\label{thm:sca}
Suppose that $n = 2k\geq 4$ is even. Then $Y_{n:n}$
stochastically dominates $Y_{n/2:n}$, that is,
\[
    {\sf P}\{Y_{n/2:n}>x\} \le {\sf P}\{Y_{n:n}>x\}
\]
\end{theorem}
\begin{proof}
Consider the system with $s=2$ where scheduling until job completion is done as follows. The system runs until one server completes the first of its two CUs, at which point it is halted. This happens at a random time distributed as $X_{1:n}$. The system of the remaining $n-1$ servers runs until one server completes the first of its 2 CUs, at which point it is halted. This happens at a random time distributed as $X_{1:(n-1)}$ measured from the moment the first server was halted. The process continues in the same manner until $k=n/2$ servers have completed the first of their 2 CUs, at which point all remaining servers are halted. This happens at a random time $T_1$ given as
\[
T_1=X_{1:n}+X_{1:(n-1)}+\dots+X_{1:(n-k+1)}
\]
At this point, the $n-k=k$ servers which have completed one CU are restarted. The job is complete when each server completes the remaining CU, which happens 
at a random time $T_2$ given as
\[
T_2=X_{1:(n-k)}+X_{1:(n-k-1)}+\dots+X_{1:1}
\]
Note that, because some servers are halted, this system cannot perform better than the 
original $s=2$ system. On the other hand, it performs as well as the $s=1$ system since we have $Y_{n:n}=T_1+T_2$.
\end{proof}

\paragraph{Numerical Analysis}
We evaluate the derived expression for $\expec[Y_{k:n}]$, and the results are shown in Fig.~\ref{fig:add}. 
\begin{figure}
    \centering
    \includegraphics[width=0.4\textwidth]{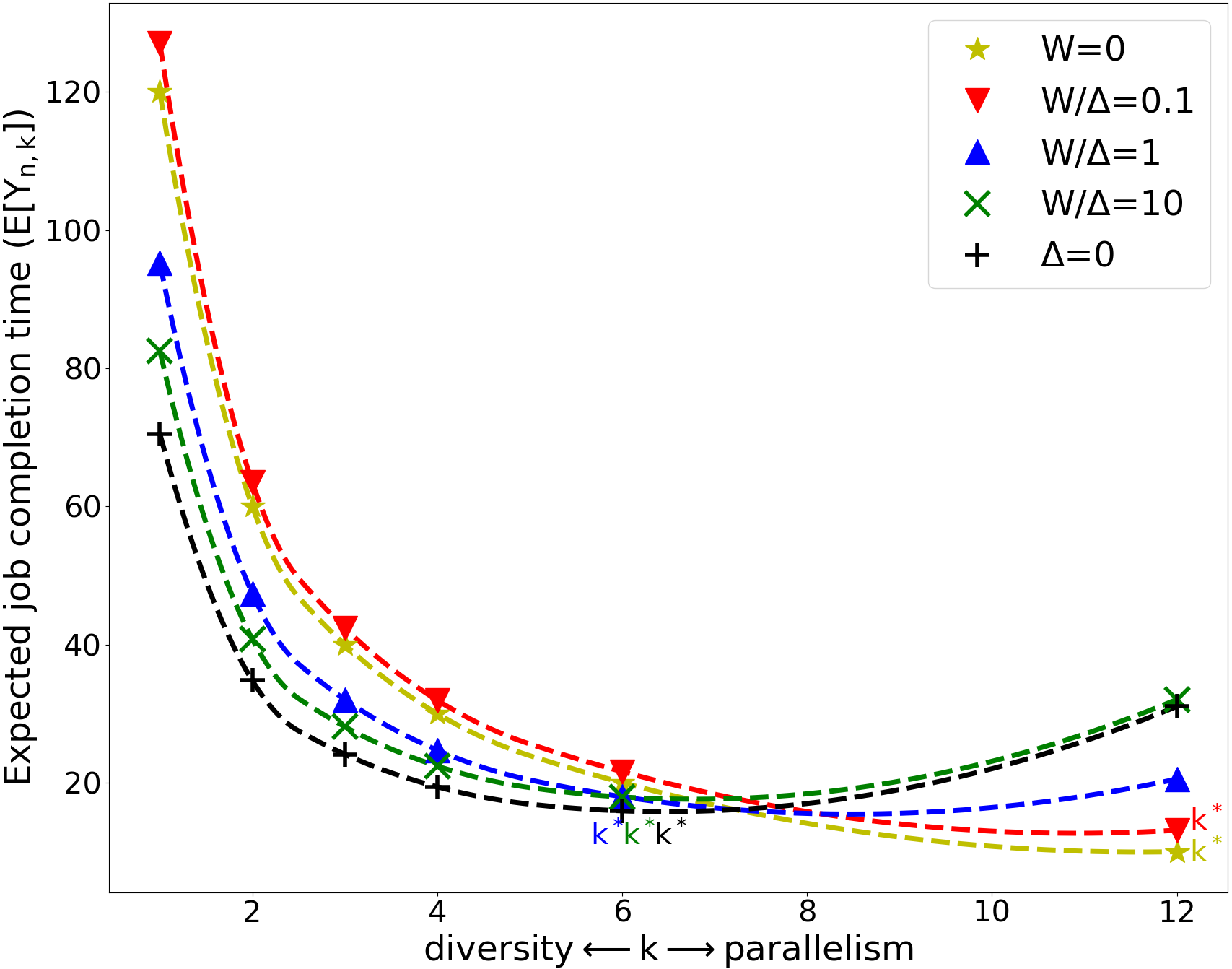}
    \caption{EExpected job completion time $\expec[Y_{k:n}]$ for (Shifted-)Exponential service time with additive scaling as a function of the diversity/parallelism parameter $k$ (cf.~\eqref{Eq:Erlang-order}). The number of workers (job size) is $n=12$, and task size per worker is $s=n/k$ (Since both $k$ and $s$ are integers, we have $k\in\{1,2,3,4,6,12\}$. We use dashed curves to connect the points corresponding to different allowed values of $k$ for a given $\Delta$ and $W$ combination.) When $W/\Delta$ is, splitting ( maximal parallelism) is the best. When $W/\Delta$ is large, there is a balance between diversity and parallelism.}
    \label{fig:add}
\end{figure}
We see that when $W/\Delta$ is small (e.g. $0$, $0.1$), splitting (maximum parallelism) gives the best performance. On the other hand, when $W/\Delta$ is large (e.g. $1$, $10$, $\infty$), we need coding in order to be optimal. The figure confirms Theorem~\ref{th:splitvsrepli} and Theorem~\ref{thm:sca} that say that splitting is better than replication and the rate half code is better than splitting when $\Delta=0$. 

Under the additive model, parallelism outperforms diversity, which was not always the case under the server-dependent and data-dependent models. Coding is optimal for some values of $\Delta$ and $W$, and the optimal code rate is around $1/2$.

\section{Pareto Service Time}
\label{sec:Pareto}
Under the Pareto model, the CU service time is given by $X$, where $X\sim \parD(\lambda, \alpha)$. The expected job completion time $\expec[Y_{k:n}]$ depends on the service time of $s=n/k$ CUs, which is determined by the service time scaling model. 
We next determine $\expec[Y_{k:n}]$ for our three scaling models.

\subsection{Server-Dependent Scaling}
\label{subsec:pareto-server}
Under the server-dependent scaling, the service time of a task consisting of $s$ CUs is given by $Y=s\cdot X$. Therefore, the job completion time is given by
\[
Y_{k:n}= s\cdot X_{k:n}, ~~\text{where}~~ X\sim \parD(\lambda, \alpha)
\]
and, by using the expression for $\expec[X_{k:n}]$ in \eqref{Eq:Pareto}, we have
\begin{equation}
\label{Eq:pareto-server}
    \expec[Y_{k:n}]= s\lambda\frac{n!}{(n-k)!}\frac{\Gamma(n-k+1-1/\alpha)}{\Gamma(n+1-1/\alpha)}
\end{equation}
The minimum expected job completion time is given by the following theorem: 
\begin{theorem}
\label{Th:Paretomethod}
The expected job completion time for $\parD(\lambda, \alpha)$ service time with server-dependent scaling reaches the minimum when $k^*$ is the ceiling or floor of $\frac{\alpha n-1}{\alpha+1}$.
\end{theorem}
\begin{proof}
From the definition of Gamma function, we have $\Gamma(z+1)=z\Gamma(z)$, and thus
\[
\expec[Y_{k:n}]=s\lambda\prod_{i=0}^{k-1}\frac{n-i}{n-1/\alpha-i}=\frac{n\lambda}{k}\prod_{i=0}^{k-1}\frac{n-i}{n-1/\alpha-i}
\]
Therefore, 
\[
\frac{\expec[Y_{k:n}]}{\expec[Y_{k+1:n}]}=\frac{(k+1)(n-1/\alpha-k)}{k(n-k)}.
\]
From this ratio, we see that 
when $k\le\frac{\alpha n-1}{\alpha+1}$, then $\expec[Y_{k:n}]\ge\expec[Y_{k+1:n}]$, and when $k\ge\frac{\alpha n-1}{\alpha+1}$, then $\expec[Y_{k:n}]\le\expec[Y_{k+1:n}]$.  
Since $k$ is an integer, the minimum $\expec[Y_{k:n}]$ is reached by setting $k$ to the ceiling or floor of $\frac{\alpha n-1}{\alpha+1}$.
\end{proof}

Pareto distribution has a finite mean only when $\alpha>1$. As $\alpha$, the tail index, decreases, the right tail of Pareto becomes heavier. From Theorem~\ref{Th:Paretomethod}, we know that the optimal $k$ ($k^*$) increases with $\alpha$. When $\alpha\downarrow1$, $k^*\approx\lceil\frac{n-1}{2}\rceil$ or $\lfloor\frac{n-1}{2}\rfloor$, and $\expec[Y_{k:n}]$ is minimized by coding. When $\alpha\rightarrow\infty$,  then $k^*\approx n$, and $\expec[Y_{k:n}]$ is minimized by splitting. Recall that $\alpha\rightarrow\infty$ implies an almost deterministic distribution, where splitting is expected to be optimal.

{\it Numerical Analysis:}
We evaluate $\expec[Y_{k:n}]$ to see how the expected job completion time changes with $k$. We consider a system with $n=12$ workers for four different values of $\alpha\in\{1.5,2,3,5\}$. We assume the Pareto scale parameter is $\lambda=1$. 
\begin{figure}
    \centering
    \includegraphics[width=0.4\textwidth]{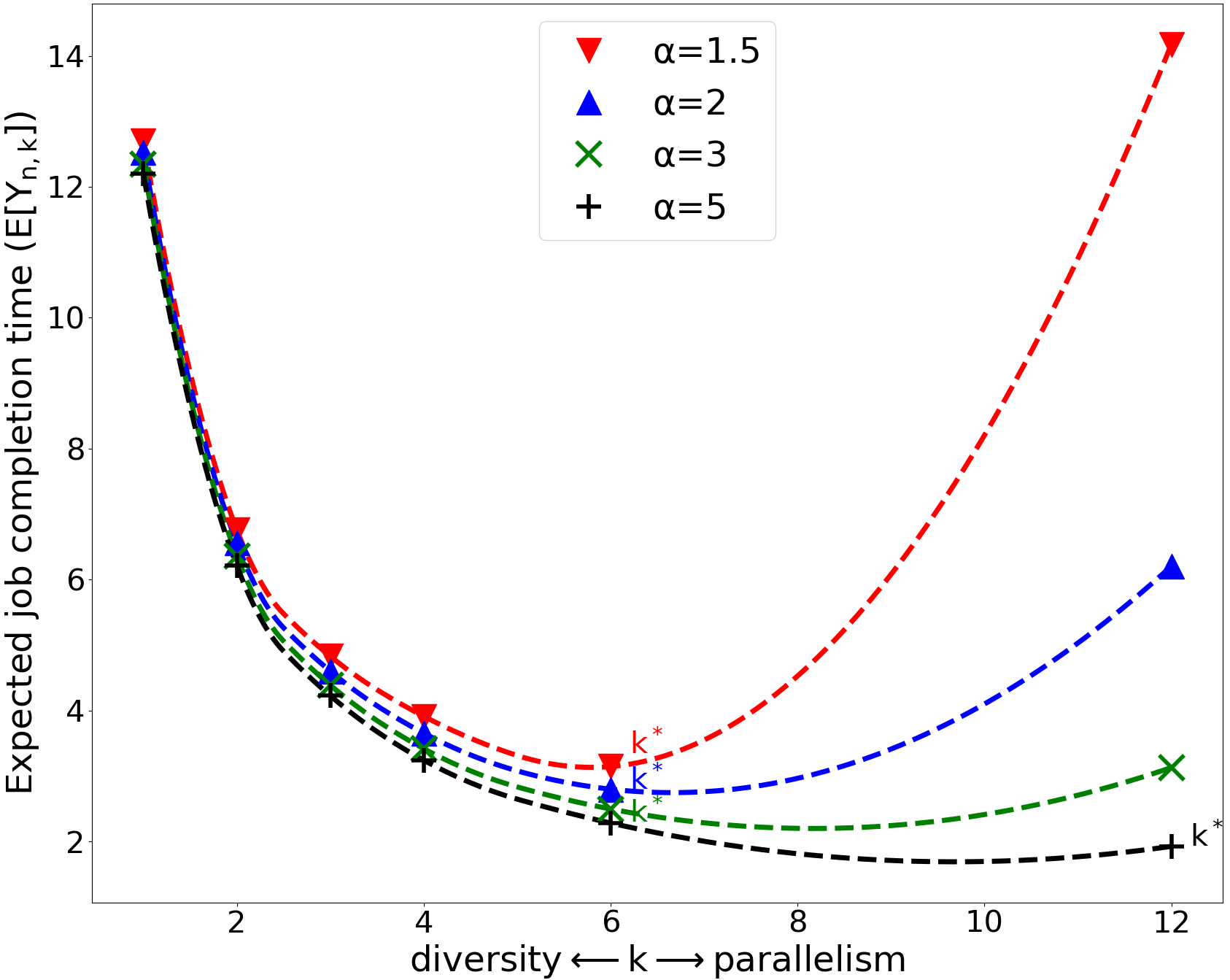}
    \caption{Expected job completion time $\expec[Y_{k:n}]$ for Pareto service time with server-dependent scaling as a function of the diversity/parallelism parameter $k$ (cf.~\eqref{Eq:pareto-server}). The number of workers (job size) is $n=12$, and task size per worker is $s=n/k$ (Since both $k$ and $s$ are integers, we have $k\in\{1,2,3,4,6,12\}$. We use dashed curves to connect the points corresponding to different allowed values of $k$ for a given $\alpha$.) The Pareto scale parameter is $\lambda=1$. Splitting or rate 1/2 coding is optimal according to different $\alpha$.}
    \label{fig:scale_Pareto}
\end{figure}
The results are plotted in Fig.~\ref{fig:scale_Pareto}. When the tail is heavy ($\alpha$=$1.5$), then $\expec[Y_{k:n}]$ reaches its minimum at $k=6$, and coding with the rate $1/2$ is optimal. Both replication and splitting have poor performance in this case. When the tail is light ($\alpha$=$5$), then $\expec[Y_{k:n}]$ is minimized by splitting. Replication still performs poorly. Otherwise ($\alpha$=$2$, $3$), coding with the rate $1/2$ is optimal. Splitting performs better than replication. From Theorem~\ref{Th:Paretomethod}, we calculate the optimal $k^*=6.8$, $7.7$, $8.8$, and $9.8$ respectively for the four scenarios. Since $k\in\{1,2,3,4,6,12\}$, $k^*$ is either $6$ or $12$. The theoretically optimal  $k^*$'s are consistent with the results in Fig.~\ref{fig:scale_Pareto}.

\subsection{Data-Dependent Scaling}
\label{Sec:data-pareto}

Under the data-dependent scaling, the service time of a task consisting of $s$ CUs is given by $Y=s\cdot\Delta+X$. Therefore, the job completion time is given by
\[
Y_{k:n}= s\cdot\Delta + X_{k:n}, ~~\text{where}~~ X\sim \parD(\lambda, \alpha)
\]
and, by using the expression for $\expec[X_{k:n}]$ in  \eqref{Eq:Pareto}, we have
\begin{equation}
\label{Eq:pareto-data}
    \expec[Y_{k:n}]= s\Delta + \lambda\frac{n!}{(n-k)!}\frac{\Gamma(n-k+1-1/\alpha)}{\Gamma(n+1-1/\alpha)}
\end{equation}
We cannot easily derive the $k^*$ that minimizes  $\expec[Y_{k:n}]$ from the above equation. However, according to the approximation for the ratio of two gamma functions given in \eqref{Ap:Pareto}, we have \[
\expec[Y_{k:n}]\approx \frac{n\Delta}{k}+\lambda\Bigl(\frac{n}{n-k}\Bigl)^{1/\alpha}.\]
Notice that the first term 
decreases with $k$, calling for maximal parallelism, whereas the second term increases with $k$, calling for maximal diversity. 

Since the service time $Y=s\cdot \Delta+X$, we can understand $Y$ as a shifted Pareto RV. Thus, $Y$ can be approximated to a constant by Pareto RV based on whether the shift $s\Delta$ is far larger or smaller than the mean of Pareto $\frac{\alpha\lambda}{\alpha-1}$. As for the Shifted-Exponential distribution with data-dependent scaling, we  conclude the following. When $\Delta\gg \frac{\alpha\lambda}{\alpha-1}$, i.e., $\expec[Y_{k:n}]\approx \frac{n\Delta}{k}$, it is optimal to operate in the maximal parallelism mode. When $\Delta \ll \frac{\alpha\lambda}{\alpha-1}$, i.e., $\expec[Y_{k:n}]\approx \lambda(\frac{n}{n-k})^{1/\alpha}$, it is optimal to operate in the maximal diversity mode. 

\paragraph*{Remark}
We have concluded here for the Pareto service time and in the previous section for the exponential service time that
replication is optimal when $\Delta\ll \expec[X]$, splitting is optimal when $\Delta\gg \expec[X]$, and coding is optimal otherwise. That conclusion applies to any service time PDF which has the first moment and the $\Delta$ and $X$ components as in the cases considered here.
The optimal code rate depends on how  $\Delta$ compares to $\expec[X]$.

\paragraph*{Numerical Analysis}
We evaluate $\expec[Y_{k:n}]$ to see how the expected job completion time changes with the values of $k$. We consider a system with $n=12$ workers for four different values of $\alpha\in\{1.5,2,3,5\}$. 
\begin{figure}
    \centering
    \includegraphics[width=0.4\textwidth]{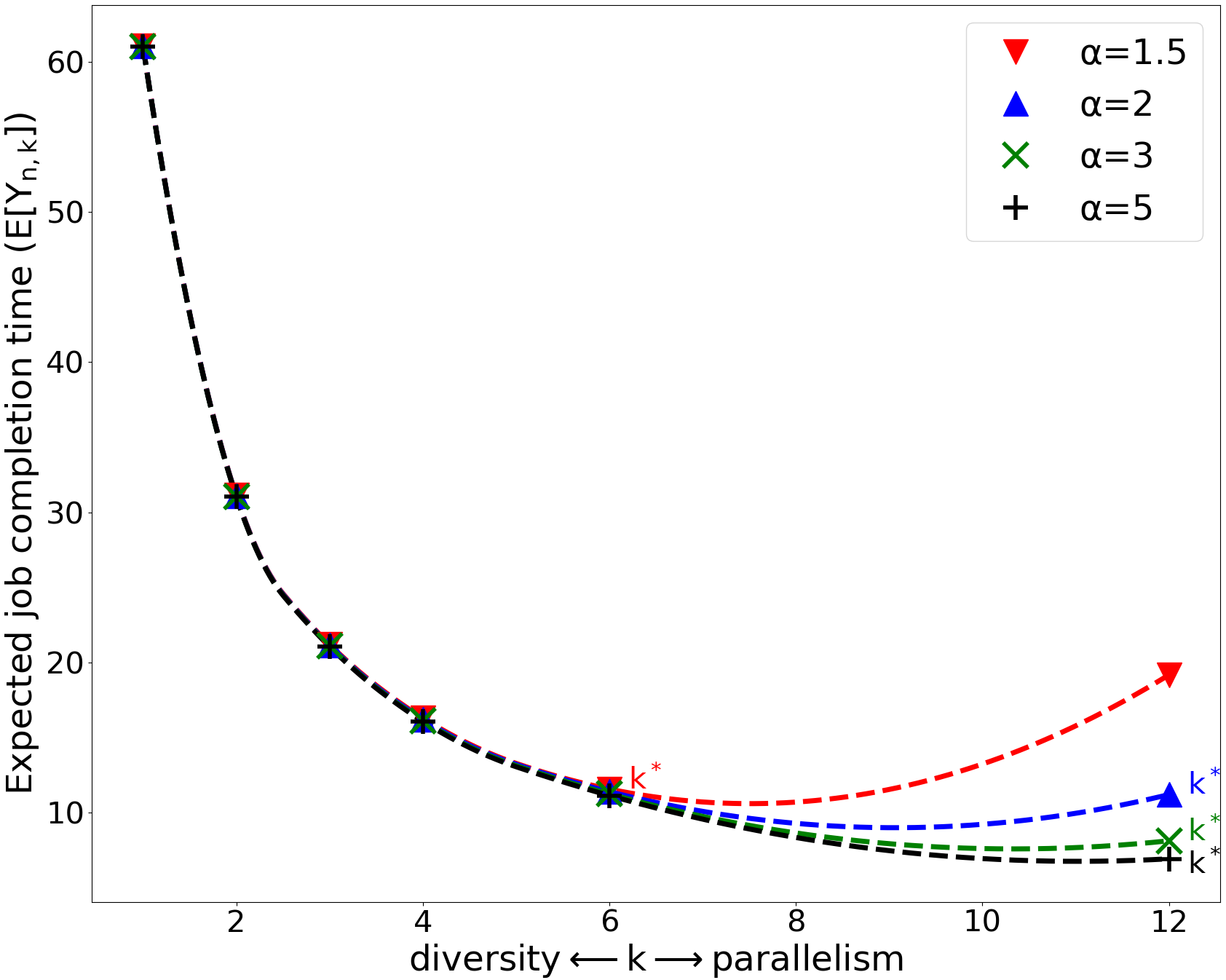}
    \caption{Expected job completion time $\expec[Y_{k:n}]$ for Pareto service time with data-dependent scaling as a function of the diversity/parallelism parameter $k$ (cf.~\eqref{Eq:pareto-data}). The number of workers (job size) is $n=12$, the shift parameter is $\Delta=5$ and task size per worker is $s=n/k$ (Since both $k$ and $s$ are integers, we have $k\in\{1,2,3,4,6,12\}$. We use dashed curves to connect the points corresponding to different allowed values of $k$ for a given $\alpha$.) The Pareto scale parameter is $\lambda=1$. Splitting or coding is optimal depending on the value of  $\alpha$.}
    \label{fig:shift_Pareto}
\end{figure}
The results are plotted in Fig.~\ref{fig:shift_Pareto}. We conclude that  splitting is optimal when the Pareto tail is light ($\alpha$ is large) and coding becomes optimal when the tail gets heavier ($\alpha$ is small).

\begin{figure}
    \centering
    \includegraphics[width=0.4\textwidth]{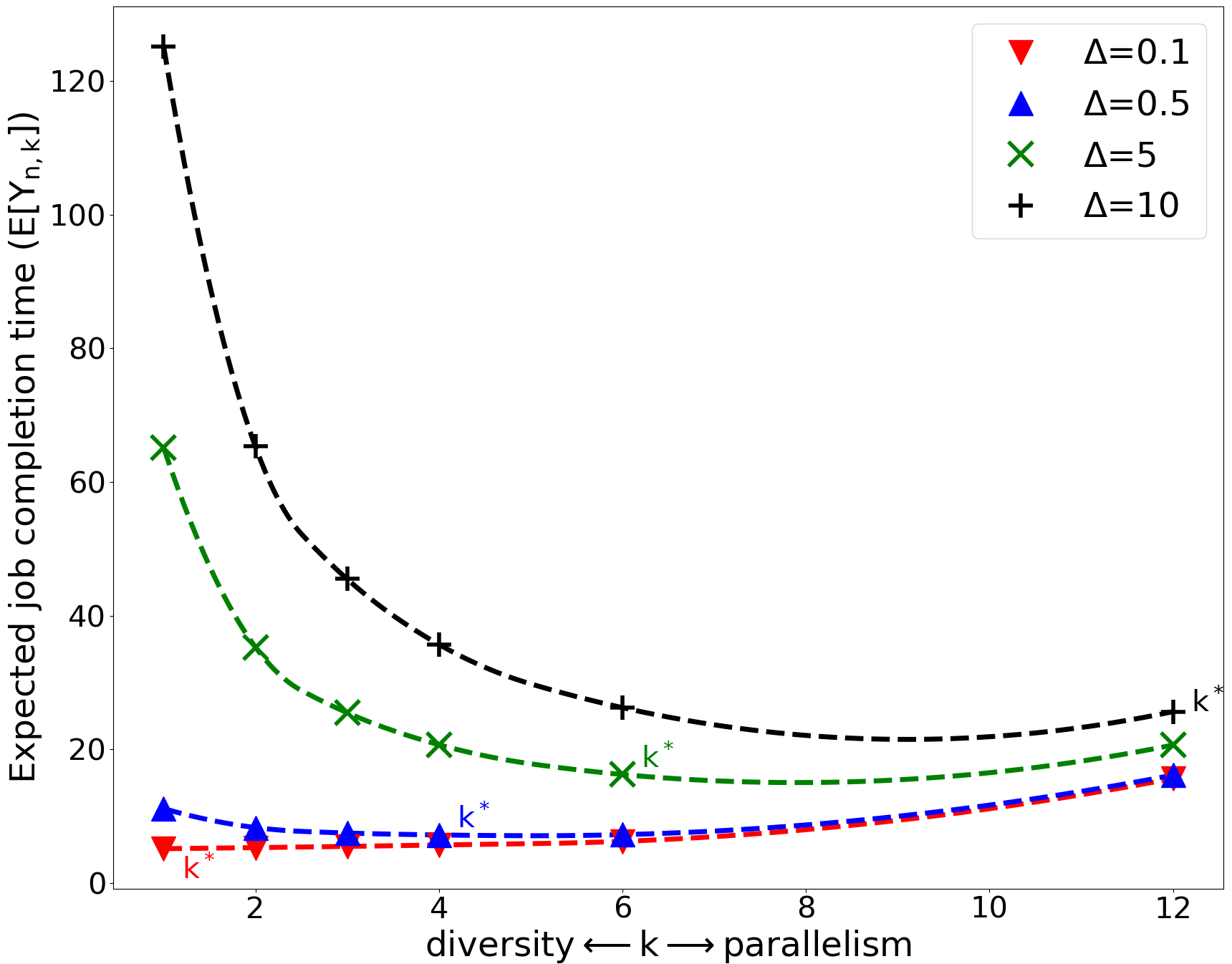}
    \caption{Expected job completion time $\expec[Y_{k:n}]$ for Pareto service time with data-dependent scaling as a function of the diversity/parallelism parameter $k$ (cf.~\eqref{Eq:pareto-data}). The number of workers (job size) is $n=12$, and task size per worker is $s=n/k$ (Since both $k$ and $s$ are integers, we have $k\in\{1,2,3,4,6,12\}$. We use dashed curves to connect the points corresponding to different allowed values of $k$ for a given $\Delta$.) The Pareto scale parameter is $\lambda=5$ and shape parameter is $\alpha=3$. The optimal code rate increases with $\Delta$.}
    \label{fig:shift_Pareto_delta}
\end{figure}
In Fig.~\ref{fig:shift_Pareto_delta}, we consider the same system for four different values of $\Delta\in\{0.1,0.5,5,10\}$. It is easy to calculate the Pareto mean $\frac{\alpha\lambda}{\alpha-1}=7.5$. When $\Delta \ll 7.5$ (e.g. $0.1$, $0.5$), replication or low-rate coding is optimal. When $\Delta$ approaches $7.5$ (e.g. $5$, $10$), splitting or high-rate coding is optimal.
These observations validate the above analysis that the optimal strategy changes with the ratio of $\Delta$ and the Pareto mean. 

\subsection{Additive Scaling}
\label{Sec:add-pareto}
Under the additive scaling, the service time of a task consisting of $s$ CUs is given by 
\[
Y=X_1+\dots +X_{s} ~~\text{where} ~~ X_i\sim \parD(\lambda, \alpha)
\]
and the job completion time is $Y_{k:n}$.

For splitting ($s=1$),  $Y$ is a Pareto RV. Thus, the job completion time $\expec[Y_{n:n}]=\expec[X_{n:n}]=\lambda n!\frac{\Gamma(1-1/\alpha)}{\Gamma(n+1-1/\alpha)}$
\cite{Pareto:Arnold15}. 
For coding and replication, the corresponding expressions are difficult to derive. Nevertheless, we can compare splitting and replication when $n$ is large by using the Law of Large Numbers (LLN). We conclude in Theorem~\ref{LLN:Pareto} that splitting outperforms replication when $n$ is sufficiently large.
\begin{theorem}
\label{LLN:Pareto}
Under the $\parD(\lambda, \alpha)$ service time with additive scaling, if the $4$-th moment exists, i.e., if $\alpha>4$, when $n$ is sufficiently large, we have $\expec[Y_{1:n}]>\expec[Y_{n:n}]$, which means splitting outperforms replication by achieving a lower expected job completion time .
\end{theorem}
\begin{proof}
We prove that $\expec[Y_{1:n}]>\expec[Y_{n:n}]$ for $n\rightarrow \infty$, by showing that there is a function $f(n)$ such that
\begin{equation}
\expec[Y_{n:n}]\le f(n) < \expec[Y_{1:n}] 
    \label{eq:fchain}
\end{equation}

We first show that there is a function $f(n)$ which satisfies $\expec[Y_{1:n}]>f(n)$ when $n\rightarrow \infty$. For the $\parD(\lambda, \alpha)$ distribution, the mean $m_{\lambda,\alpha}=\lambda\alpha/(\alpha-1)$ exist if $\alpha>1$. Therefore, by the LLN, we have $\frac{1}{n}\sum_{j=1}^{n}X_j\rightarrow m_{\lambda,\alpha} \ as \ n\rightarrow  \infty $. 

If $\alpha>4$, the $4$-th moment of Pareto distribution exists. Let $\expec[X^4]=\xi<\infty$ and $Z=X-m_{\lambda,\alpha}$. By applying Jensen's inequality and removing negative terms, it follows that $\expec[Z^4]=\expec[X^4 - 4X^3 m_{\lambda,\alpha} + 6 X^2 m_{\lambda,\alpha}^2 -4m_{\lambda,\alpha}^3 X] +m_{\lambda,\alpha}^4=\expec[X^4 + 6 X^2 m_{\lambda,\alpha}^2 ] - 4\expec[X^3] m_{\lambda,\alpha} -3 m_{\lambda,\alpha}^4\le 7\xi$. Next, define an RV $S$ to be $S=\sum_{i=1}^{n} Z_i$. Since $Z_i$ are i.i.d. (for $i=1,2,\dots,n$) and $\expec[Z_i]=0$, by expanding $S^4$, we know the terms which have 4 or 3 different indices have 0 expectation. For example, $\expec[Z_1^2Z_2Z_3]=\expec[Z_1^2]\expec[Z_2]\expec[Z_3]=0$. The terms which have 2 different indices have the expectation $\expec[Z_i^2Z_j^2]\le \xi$ by Jensen's inequality. And the coefficient is $3n^2-3n$. Similarly, the terms which have 1 index have the expectation $\expec[Z_i^4]\le \xi$ and the coefficient $n$. Thus we have $\expec[S^4]\le3\cdot 7\xi\cdot n^2$. 
Furthermore, by Markov's inequality, we obtain
\begin{equation}
 \label{eq:p1}
    p_n={\Pr}\{|S|/n\ge\eta\}\le \frac{\expec[S^4]}{n^4\eta^4}\le \frac{21\xi}{n^2\eta^4},
\end{equation}
where $|S|$ is the absolute value of $S$, and $\eta$ is a small positive number.
Note that $S+nm_{\lambda,\alpha}=\sum_{i=1}^{n} X_i$ is the time a worker takes to complete his task  under the replication strategy. Let 
$S_{i}+nm_{\lambda,\alpha}$ be the time that the $i$-th worker takes to complete his task. Then, the job completion time is $Y_{1:n}=nm_{\lambda,\alpha}+\min_i S_i$. Since $S_{1},\dots, S_{n}$ are independent, we have that $
    {\Pr}\{Y_{1:n}> n(m_{\lambda,\alpha}-\eta)\}={\Pr}\{Y_{1:n}-nm_{\lambda,\alpha}> -n\eta)\}=({\Pr}\{S_{1}>-n\eta\})^n
    > (1-{\Pr}\{|S|/n\ge \eta\})^n $.

By using the bound in \eqref{eq:p1}, it follows that 
\begin{equation}
\label{eq:p3}
    {\Pr}\{Y_{1:n}> n(m_{\lambda,\alpha}-\eta)\}>(1-p_n)^n\ge \Bigl(1-\frac{21\xi}{n^2\eta^4}\Bigr)^n
\end{equation}
When $n\rightarrow \infty$, the rightmost term above tends to $1$. Therefore, when $n$ is sufficiently large, we have that ${\Pr}\{Y_{1:n}> n(m_{\lambda,\alpha}-\eta)\})\rightarrow 1$, and thus 
 $\expec[Y_{1:n}]> n(m_{\lambda,\alpha}-\eta)$. Let $f(n)=n(m_{\lambda,\alpha}-\eta)$.
 
We next prove that $f(n)=n(m_{\lambda,\alpha}-\eta)$ satisfies $f(n)>\expec[Y_{n:n}]$ for $n\rightarrow \infty$.
The job completion time for splitting is 
$Y_{n:n}=\max_{1\le i\le n} X_i$.
Since $\alpha>4$, the second moment of Pareto distribution exists. Let $\expec[X^2]=\zeta<\infty$. Then, by Markov's inequality, we have
\begin{equation}
 \label{eq:p4}
    q_n={\Pr}\{X\ge n(m_{\lambda,\alpha}-\eta)\}\le \frac{\zeta}{n^2(m_{\lambda,\alpha}-\eta)^2}
\end{equation}
Since $X_{1},\dots, X_{n}$ are independent, we have that $
    {\Pr}\{Y_{n:n}< n(m_{\lambda,\alpha}-\eta)\}=({\Pr}\{X_{1}<n(m_{\lambda,\alpha}-\eta)\})^n= (1-{\Pr}\{X\ge n(m_{\lambda,\alpha}-\eta)\})^n$.
By using the bound in \eqref{eq:p4}, it follows that $
    {\Pr}\{Y_{n:n}< n(m_{\lambda,\alpha}-\eta)\}=(1-q_n)^n\ge \Bigl(1-\frac{\zeta}{n^2(m_{\lambda,\alpha}-\eta)^2}\Bigr)^n$
    
When $n\rightarrow \infty$, the rightmost term above tends to $1$. Therefore, when $n$ is sufficiently large, we have ${\Pr}\{Y_{n:n}< n(m_{\lambda,\alpha}-\eta)\}\rightarrow 1$, and thus $\expec[Y_{n:n}]< n(m_{\lambda,\alpha}-\eta)=f(n)$. 

We have shown that when $n$ is sufficiently large, the inequalities \eqref{eq:fchain} hold, which proves the theorem. 
\end{proof}

From the proof of Theorem~\ref{LLN:Pareto}, we get the lower bound on the expected time for replication: $\expec[Y_{1:n}]\ge n(m_{\lambda,\alpha}-\eta)r_n$. Here $\eta$ is a small and positive and $r_n=(1-\frac{21\xi}{n^2\eta^4})^n$ by \eqref{eq:p3}.

Observe that having Pareto service time plays no special role, and the arguments we used apply to any PDF which has the $4$-th moment, as we formally express in Corollary~\ref{Le:general}.
\begin{corollary}
\label{Le:general}
For a general service time distribution with the fourth moment, splitting results in a smaller expected job completion time than replication under additive scaling when the number of workers is sufficiently large. 
\end{corollary}

\paragraph*{Simulation Analysis}
Although the expression of $\expec[Y_{k:n}]$ is unknown, we can analyze $\expec[Y_{k:n}]$ vs.\ $k$ by simulation. For each worker, we sum $s$ Pareto RVs samples. We then compare the $n$ workers' service times to get $Y_{k:n}$. We estimate $\expec[Y_{k:n}]$ by calculating the average of $10000$ values of $Y_{k:n}$. We consider a system with $n=12$ workers and four different values of the tail index $\alpha\in\{1.3,2,3,5\}$.  
The simulation results are plotted in Fig.~\ref{fig:add_Pareto}.
\begin{figure}
    \centering
    \includegraphics[width=0.4\textwidth]{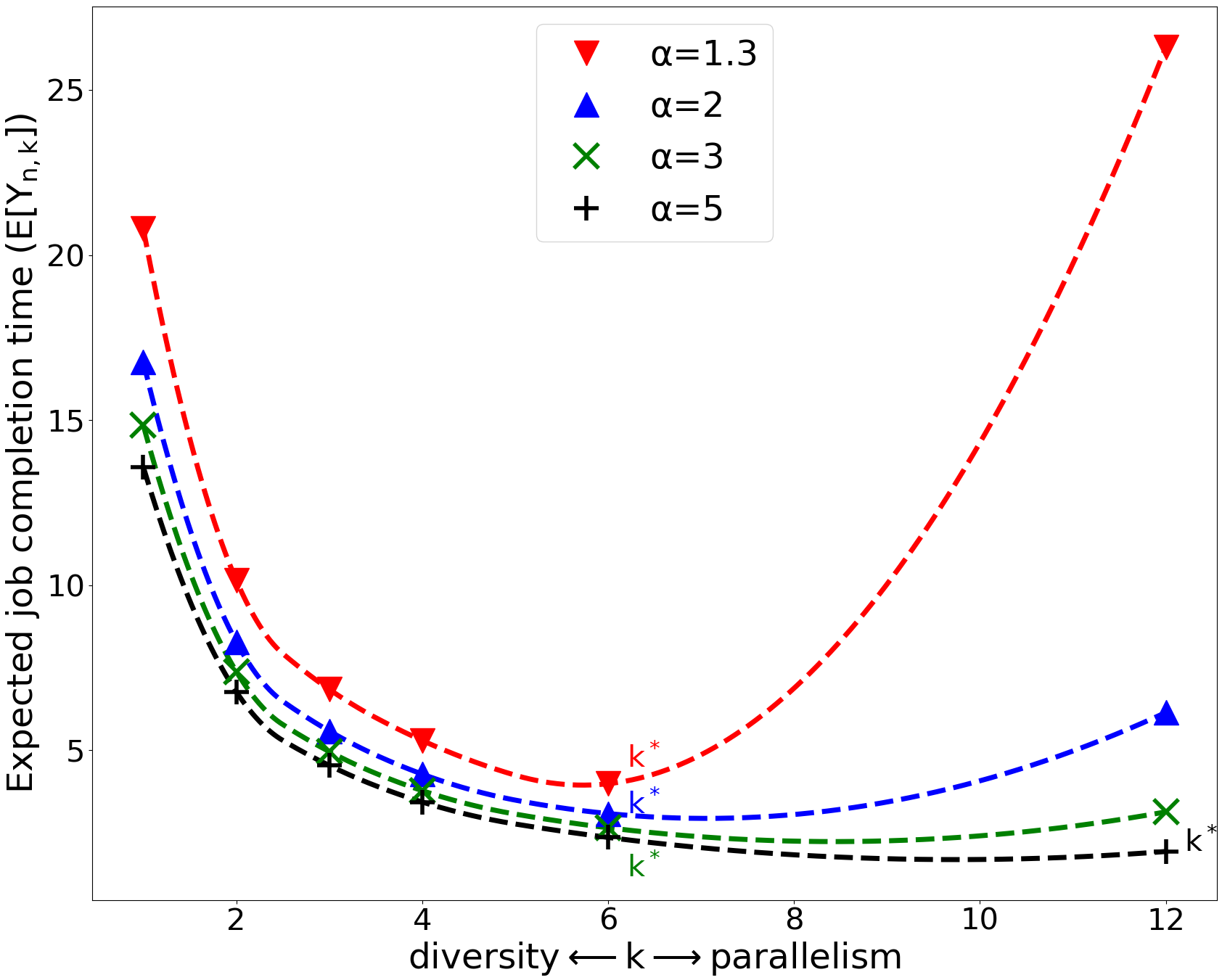}
    \caption{Expected job completion time $\expec[Y_{k:n}]$ for Pareto service time with additive scaling as a function of the diversity/parallelism parameter $k$. The number of workers (job size) is $n=12$, and task size per worker is $s=n/k$ (Since both $k$ and $s$ are integers, we have $k\in\{1,2,3,4,6,12\}$. Results are obtained by simulation. We use dashed curves to connect the points corresponding to different allowed values of $k$ for a given $\alpha$.) The Pareto scale parameter is $\lambda=1$. Splitting or coding is optimal depending on the value of $\alpha$.}
    \label{fig:add_Pareto}
\end{figure}
 We observe that splitting is optimal for a light tail (large $\alpha$), and coding is optimal for a heavy tail (small $\alpha$). When coding is optimal, the optimal code rate is close to $1/2$. (Recall that we do not consider all fractions as possible code rates.)

\begin{figure}
    \centering
    \includegraphics[width=0.4\textwidth]{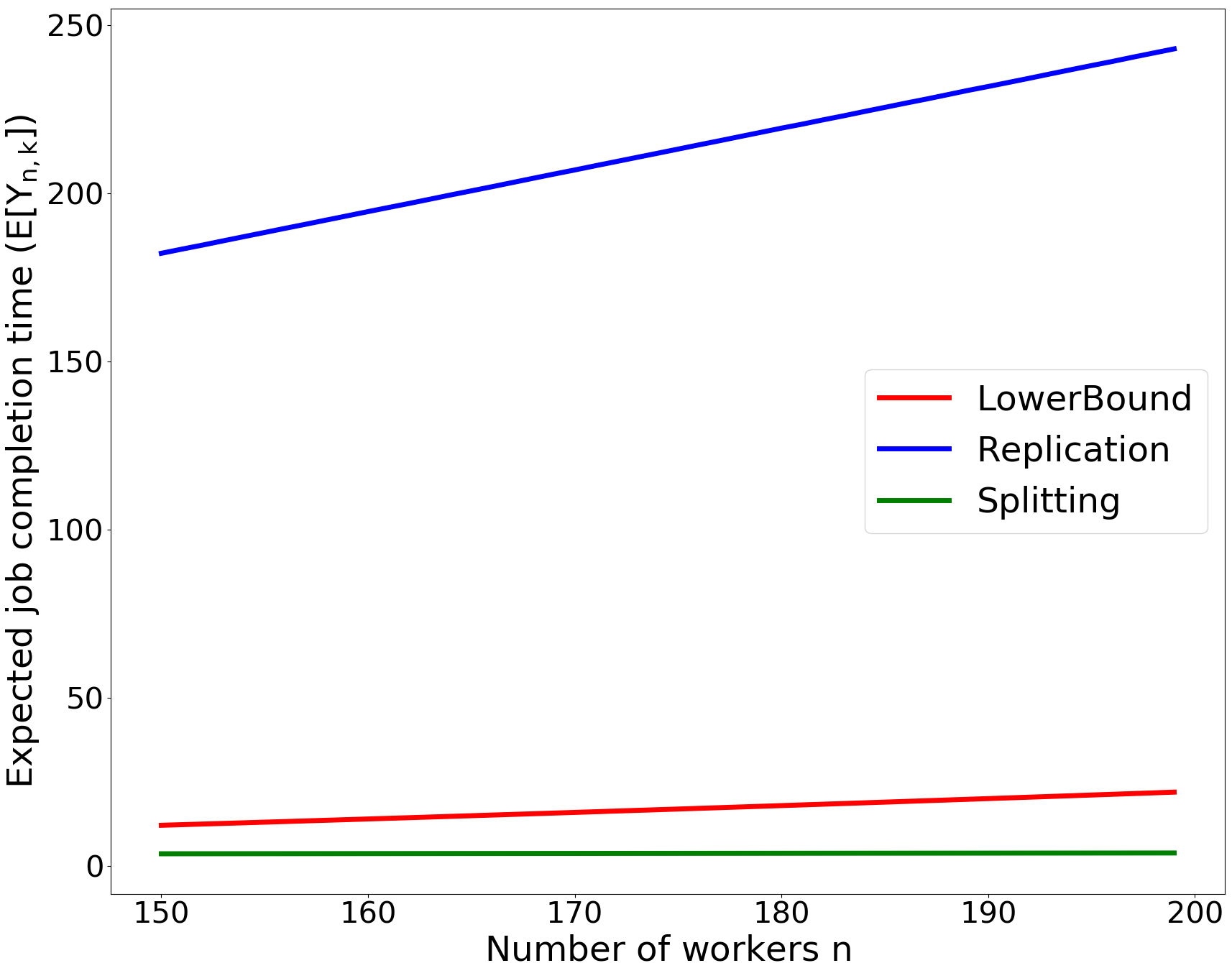}
    \caption{Expected job completion time $\expec[Y_{k:n}]$ for Pareto service time with additive scaling as a function of the number of workers (job size) $n$.  The Pareto scale parameter is $\lambda=1$ and the shape parameter is $\alpha=4.5$. We set $\eta=1$.  Replication results are obtained by simulation. For Splitting and Lower Bound, we use \eqref{Eq:Pareto} and \eqref{eq:p3}, respectively. The lower bound on replication shows that splitting outperforms replication. }
    \label{fig:add_Pareto_LLN}
\end{figure}
In Fig.~\ref{fig:add_Pareto_LLN}, we compare replication and its lower bound with splitting. The result for replication is from simulation, and the results for splitting and the lower bound are from their expressions. We observe that $\expec[Y_{1:n}]$ with replication is much larger than its lower bound. Nevertheless, the lower bound is clearly larger than $\expec[Y_{n:n}]$ with splitting. Thus we conclude that splitting outperforms replication by achieving a lower expected job completion time for the large $n$ scenario. 

\section{Bi-Modal Service Time}
 \label{sec:Bi-Modal}
 Under the Bi-Modal model, the CU service time is given by $X$, where $X\sim \BMD(B, \epsilon)$ as given in \eqref{eqn:toy_X_dist}. The expected job completion time $\expec[Y_{k:n}]$ depends on the service time of $s=n/k$ CUs, which is determined by service time scaling model. 
In the following three subsections, we determine $\expec[Y_{k:n}]$ for our three scaling models.

\subsection{Server-Dependent Scaling}
Under the server-dependent scaling, the service time of a task consisting of $s$ CUs is given by $Y=s\cdot X$. Therefore, the job completion time is given by
\[
Y_{k:n}= s\cdot X_{k:n}, ~~\text{where}~~X \sim \BMD(B, \epsilon)
\]
It is easy to see that $X_{k:n}$ is a Bi-Modal random variable
\begin{align*}
X_{k:n} = \begin{cases}
B ~~&\text{w.p.} ~~ \textstyle{\sum_{i=0}^{k-1}\binom{n}{i}(1-\epsilon)^i\epsilon^{n-i}}\\
1 ~~&\text{w.p.} ~~ 1 - \textstyle{\sum_{i=0}^{k-1}\binom{n}{i}(1-\epsilon)^i\epsilon^{n-i}}
\end{cases} 
\end{align*}
and that its expectation is given by
\begin{equation}
\label{Eq:bimodal_server}
    \expec[Y_{k:n}]= s+s(B-1)\sum^{k-1}_{i=0} \binom{n}{i} \epsilon^{n-i}(1-\epsilon)^{i}.
\end{equation}

From the definition of the $\BMD(B, \epsilon)$ distribution 
given in \eqref{eqn:toy_X_dist}, we see that when the probability of  straggling $\epsilon$ is small ($\epsilon\rightarrow 0$), then $X$ is highly concentrated around $1$, and when $\epsilon$ is large ($\epsilon\rightarrow 1$), then $X$ is highly concentrated around $B$.  Therefore, in these two extreme cases, we have little variation in $X$, and thus diversity at the expense of parallelism does not help. Therefore, $\expec[Y_{k:n}]$ reaches its minimum at $k=n$, which means that splitting is optimal.

Another case, in which we also have little variation in $X$, is when the magnitude of straggling $B$ is small. We formally show that splitting is optimal in Proposition~\ref{Th:Bi-Modal} for $B \le 2$. 

\begin{proposition}
\label{Th:Bi-Modal}
 For $\BMD(B,\epsilon)$ service time with server-dependent scaling, if $B \le 2$, the expected job completion time $\expec[Y_{k:n}]$ reaches its minimum at $k=n$ (maximal parallelism).
\end{proposition}
\begin{proof}
Since $k$ is an integer that divides $n$, we know that either $k=n$ or $k\le n/2$. 
When $k=n$, then $\expec[Y_{n:n}]= 1+(B-1)\sum^{n-1}_{i=0} \binom{n}{i} \epsilon^{n-i}(1-\epsilon)^{i}<B \le 2$.
When $k \le n/2$, then $\expec[Y_{k:n}]> s = n/k \ge 2$.
\end{proof}

When $B>2$, computing $k$ that minimizes the expression for $\expec[Y_{k:n}]$ in \eqref{Eq:bimodal_server} becomes harder, and we use the LLN to approximate $\expec[Y_{k:n}]$ for large $n$, as follows:
\begin{theorem}
\label{Thm:bimodal_LLN}
For the $\BMD(B, \epsilon)$, where $B>2$, service time with server-dependent scaling, we have
\begin{equation}
    \expec[{Y_{k:n} }] \sim\frac{1}{r} p_r + \frac{B}{r} q_r~~\text{as}~~n\rightarrow \infty
    \label{eq:BMD-LLN}
\end{equation}
Here, $r=k/n$ is the code rate, $p_r\rightarrow 1$ if $ 1 - \epsilon > r$ and $p_r\rightarrow 0$ if $ 1 - \epsilon < r$, and $q_r=1-p_r$. 
\end{theorem}
\begin{proof}
At server $i$, $i=1,\dots, n$, the task completion time $Y_i$ is a Bi-Modal random variable taking value $s\cdot 1$ or $s\cdot B$. Therefore, the job completion time $Y_{k:n}$ is also a Bi-Modal random variable taking value $s$ or $sB$. We define an indicator function $\mathbf{1}_{\{Y_i|\{s\}\}}:\{Y_i|\{s,sB\}\}\rightarrow \{0,1\}$, which takes value one when $Y_i$ takes value $s$ and zero otherwise. Let $M$ be the sum of $n$ i.i.d. indicators $\mathbf{1}_{\{Y_i|\{s\}\}}$ where $i=1,\dots, n$. ($M$ is the number of servers whose completion time took value $s$ in a given realization.)
Then $
    {\Pr}\{Y_{k:n}= s\} = {\Pr}\{M\ge k\} ~~ and~~ {\Pr}\{Y_{k:n}=sB\} = {\Pr}\{M<k\}.$
    
We next look into $M$ through the LLN lens. Let $r \doteq k/n$ and $M_n \doteq M/n$. Observe that $M_n \rightarrow 1 - \epsilon$ as $n \rightarrow \infty$.  We see that $p_r = {\Pr} \{ Y_{k:n} = s\}  = {\Pr} \{M_n \geq r\} \rightarrow 1$ if $ 1 - \epsilon > r$ and 0 if the inequality is reversed. Since $s = 1/r$, the LLN approximation yields the expression \eqref{eq:BMD-LLN}.

\end{proof}

From \eqref{eq:BMD-LLN}, we see that $\expec[{Y_{k:n} }]$ is a convex, unimodal function of $r$ on $[0,1-\epsilon]$, and a decreasing function of $r$ on $(1-\epsilon,1]$. Therefore, $\expec[{Y_{k:n} }]$ has two local minimums: $1/(1 - \epsilon)$ at $r=1-\epsilon$ and $B$ at $r=1$, which we compare and conclude the following. Whether coding or splitting  is optimal depends on how the probability of straggling $\epsilon$ and straggling magnitude 
$B$ compare to each other.  When $\epsilon \le (B-1)/{B}$, the global minimum is $1/(1 - \epsilon)$, and thus coding with the code rate $r=1-\epsilon$ is optimal. When $\epsilon > (B-1)/{B}$, the global minimum is $B$, and thus splitting is optimal. 

Notice that the LLN based approximation provides both qualitative and quantitative insights. We obtain a useful approximation to the optimal code rate in \eqref{eq:BMD-LLN}, which is much simpler and insightful than the exact expression \eqref{Eq:bimodal_server}. Moreover, from the numerical analysis below, we can observe that the approximation is useful even for small $n$.

\paragraph*{Remark}
Based on the insight we gained from the above analysis, we make the following conjecture about the optimal strategy for a general class of service time PDFs. The claim in the conjecture holds for the PDFs considered in this paper.
\begin{conjecture}
\label{conject:server}
Under server-dependent scaling and a CU service time $X=\delta+Z$ where $\delta\ge 0$ is a constant and $Z$ is a random variable whose support includes $0$,  the expected job completion time is minimized by replication when  $\delta=0$ and by coding or splitting otherwise. 
\end{conjecture}

\paragraph*{Numerical Analysis}

\begin{figure}
    \centering
    \includegraphics[width=0.4\textwidth]{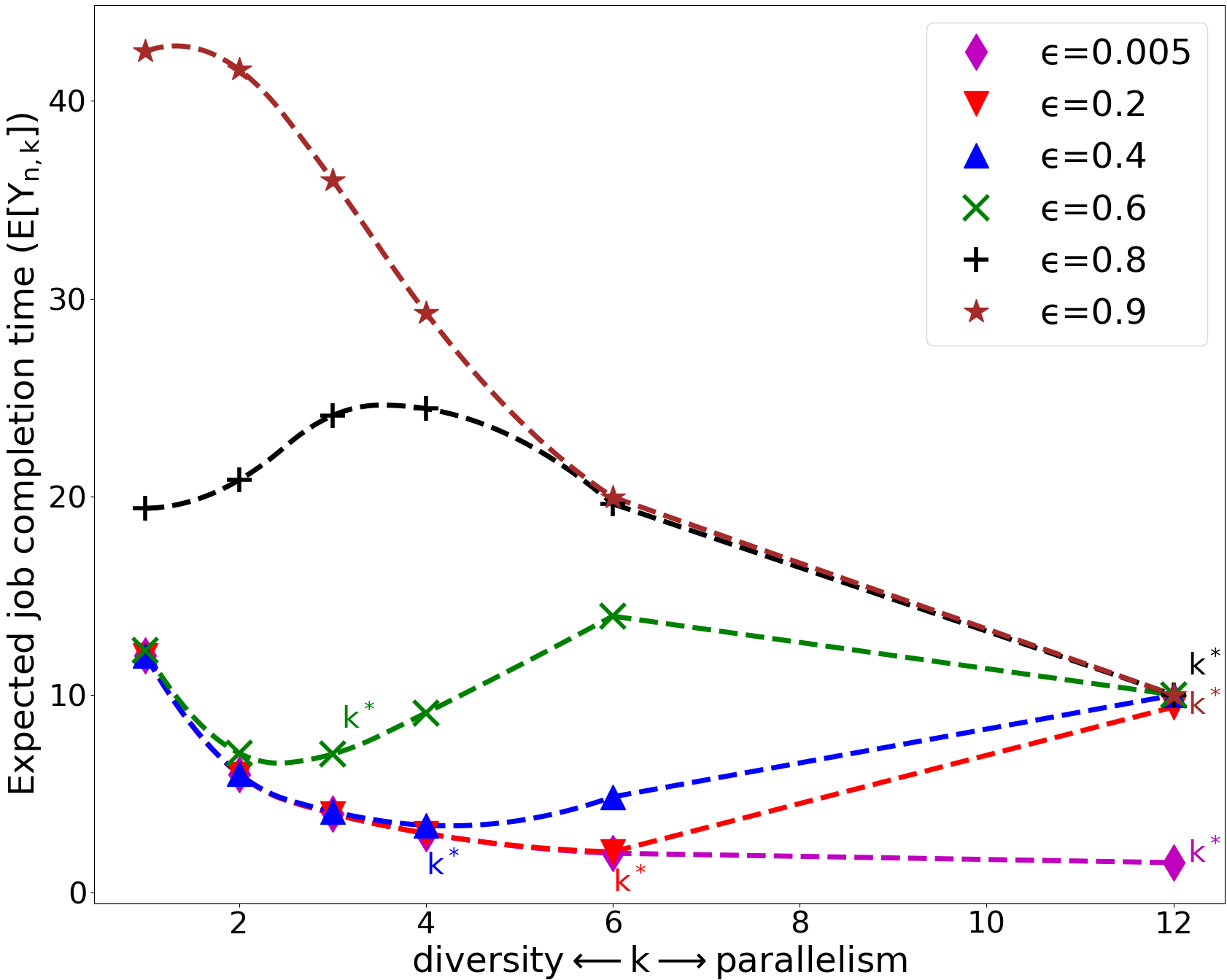} 
    \caption{Expected job completion time $\expec[Y_{k:n}]$ for Bi-Modal service time with server-dependent scaling as a function of the diversity/parallelism parameter $k$ (cf.~\eqref{Eq:bimodal_server}).  The straggling magnitude $B$ is $10$. The number of workers (job size) is $n=12$, and task size per worker is $s=n/k$ (Since both $k$ and $s$ are integers, we have $k\in\{1,2,3,4,6,12\}$. We use dashed curves to connect the points corresponding to different allowed values of $k$ for a given $\alpha$.) The optimal code rate decreases as $\epsilon$ increases, except when $\epsilon$ approaches $1$, where the maximal parallelism is optimal.}
    \label{fig:server_dep_e}
\end{figure}
In Fig.~\ref{fig:server_dep_e}, we evaluate the expression of $\expec[Y_{k:n}]$ to see how the expected job completion time changes with the diversity/parallelism parameter $k$. We consider a system with $n=12$ workers for six different values of $\epsilon\in\{0.005,0.2,0.4,0.6,0.8,0.9\}$. 
Some observations can be made from the figure:
when $\epsilon\rightarrow0$ (e.g. $0.005$), $\expec[Y_{k:n}]$ decreases with $k$, and splitting is optimal. When $\epsilon$ is small (e.g. $0.2$, $0.4$, $0.6$), $\expec[Y_{k:n}]$ reaches its minimum at $k\in[2,6]$, thus coding is optimal and the optimal code rate decreases with increasing $\epsilon$. When $\epsilon$ is large (e.g. $0.8$, $0.9$), $\expec[Y_{k:n}]$ reaches its minimum at $k=12$, and splitting is optimal. From these observations, we conclude that as $\epsilon$ increases, $\expec[Y_{k:n}]$ is minimized by introducing more diversity. However, when $\epsilon$ approaches $1$, $X$ approaches to a deterministic random variable, then $\expec[Y_{k:n}]$ is minimized by maximal parallelism.

\begin{figure}
    \centering
    \includegraphics[width=0.4\textwidth]{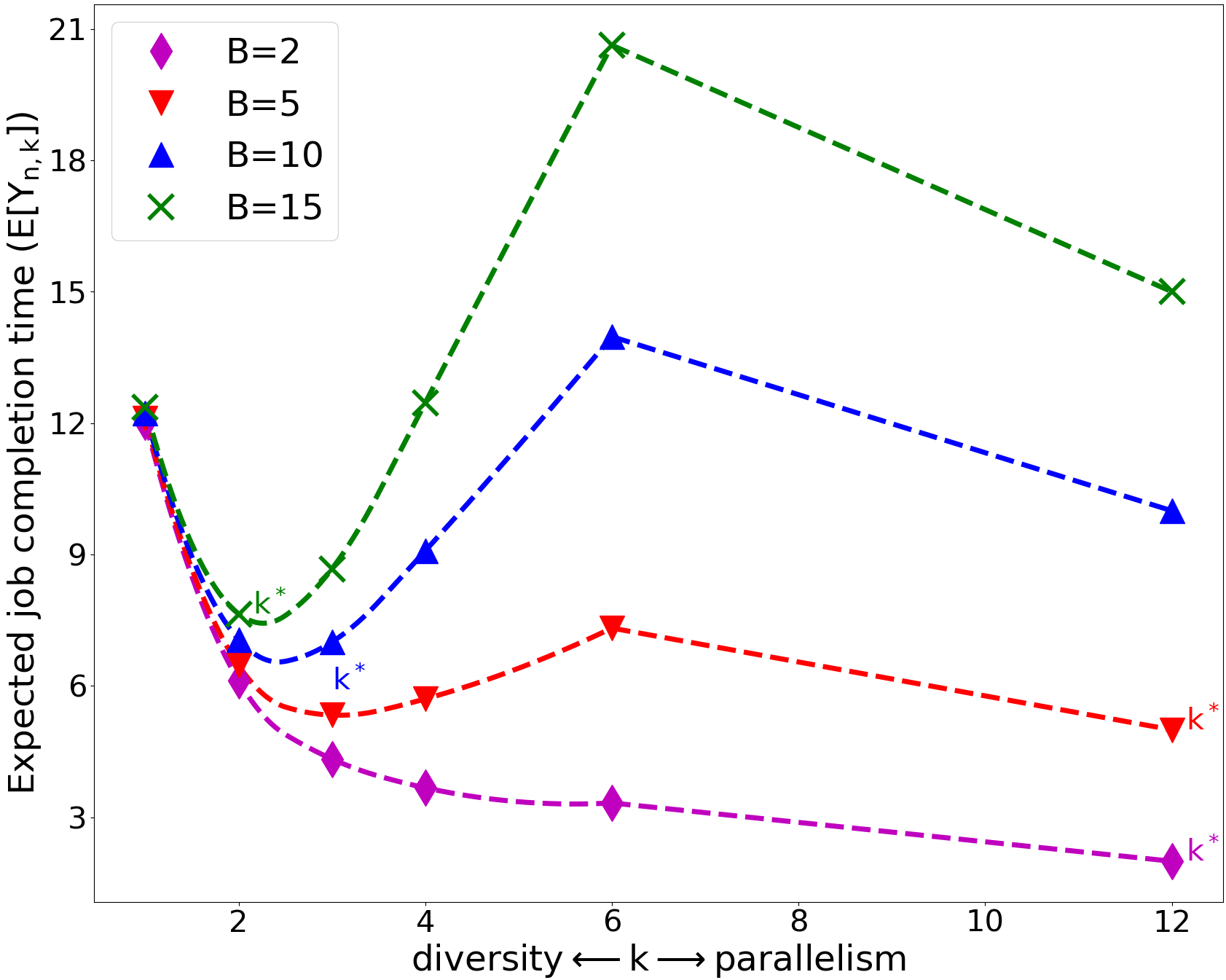}
    \caption{Expected job completion time $\expec[Y_{k:n}]$ for Bi-Modal service time with server-dependent scaling as a function of the diversity/parallelism parameter $k$ (cf.~\eqref{Eq:bimodal_server}).  The straggling probability $\epsilon$ is $0.6$.  The number of workers (also job size) is $n=12$, and task size per worker is $s=n/k$ (Since both $k$ and $s$ are integers, we have $k\in\{1,2,3,4,6,12\}$. We use dashed curves to connect the points corresponding to different allowed values of $k$ for a given $\alpha$. ) The optimal code rate decreases with increasing $B$. }
    \label{fig:server_dep_B}
\end{figure}
In Fig.~\ref{fig:server_dep_B}, we evaluate $\expec[Y_{k:n}]$ vs. $k$ for four different values of $B$ (from 2 to 15). We observe that
when $B$ is small (e.g., 2, 5), splitting is optimal. When $B$ is large (e.g. 10, 15), $\expec[Y_{k:n}]$ is minimized by coding and the optimal code rate increases with $B$. When $B$ is very large (e.g. $B>150$), not shown in Fig.~\ref{fig:server_dep_B}, replication is optimal.
From the above, we conclude that the magnitude of $B$ determines the diversity/parallelism trade-off: when $B$ is small, we gain more from parallelism, s.t. splitting. When $B$ is large, we gain more from diversity, s.t. coding or replication.

\begin{figure}[hbt]
	\centering
	\includegraphics[width=0.4\textwidth]{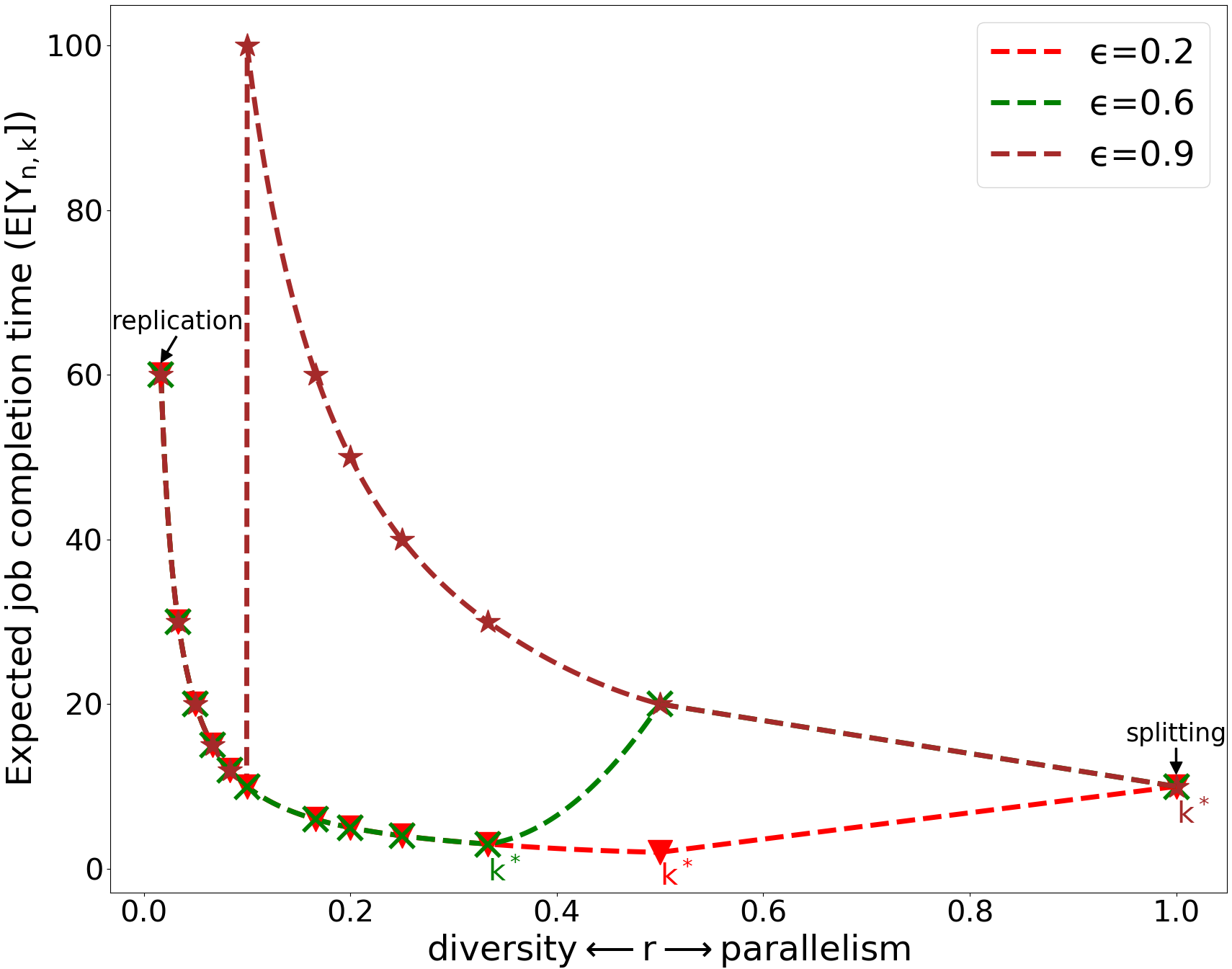}~\\
	\includegraphics[width=0.4\textwidth]{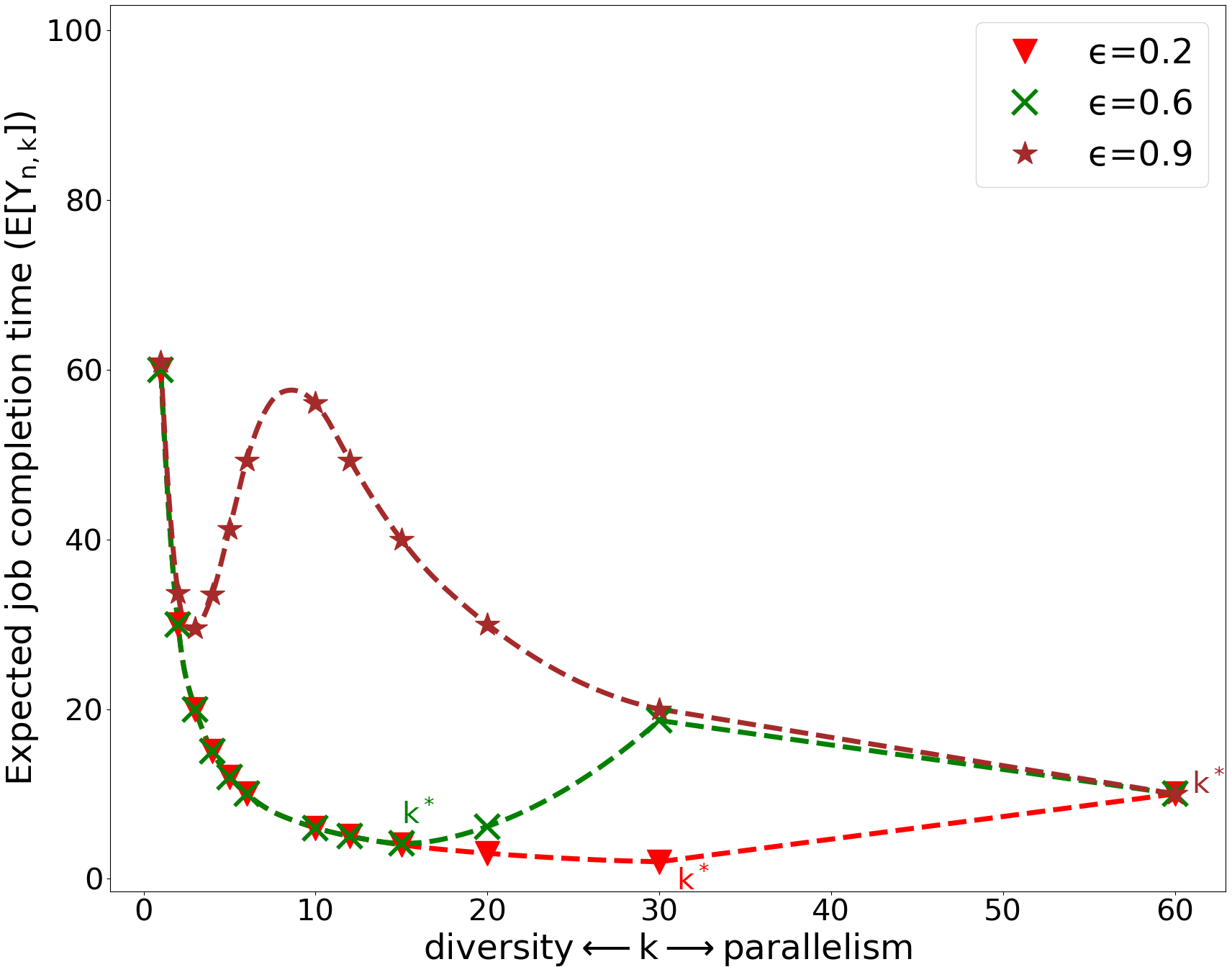}~
	\caption{Comparisons between the LLN approximation of $\expec[Y_{k:n}]$ given by \eqref{eq:BMD-LLN} and the exact result given by \eqref{Eq:bimodal_server} for Bi-Modal service time with server-dependent scaling. The straggling magnitude $B$ is $10$. The number of workers (job size) is $n=60$, task size per worker is $s=n/k$. (upper) The LLN approximation shows the $\expec[Y_{k:n}]$ vs.\ the code rate $r=k/n$. (lower) The exact dependence of $\expec[Y_{k:n}]$ on the diversity/parallelism parameter $k$. (Since both $k$ and $s$ are integers, we have $k\in\{1,2,3,4,5,6,10,12,15,20,30,60\}$. We use dashed curves to connect the points corresponding to different allowed values of $k$ for a given $\alpha$.) The LLN approximation performs well here. }
	\label{fig:server_dep_LLN}
\end{figure}
In Fig.~\ref{fig:server_dep_LLN}, we compare the LLN approximation of $\expec[Y_{k:n}]$  with the exact result \eqref{Eq:bimodal_server}. The upper graph shows the LLN approximation of the dependence of $\expec[Y_{k:n}]$ on $r=k/n$,  and the lower graph shows the exact dependence of $\expec[Y_{k:n}]$ on $k$. There are $n=60$ workers and three different values of $\epsilon$. Since $k$ and $s$ are integers, we can only have $k\in\{1,2,3,4,5,6,10,12,15,20,30,60\}$. The corresponding $\expec[Y_{k:n}]$ are marked in the figure.
To evaluate the approximation, we compare three important metrics: the local minimums, the optimal $k^*$ and the minimum $\expec[Y_{k:n}]$, and make the following observations. First, for each value of $\epsilon$, both LLN and the exact result have the same number of local minimums. However, when $\epsilon=0.9$, the LLN gives a much smaller value of the first local minimum. Second, when $\epsilon=0.2$ and $0.9$, the LLN shows the same values of $k^*$'s as the exact result. When $\epsilon=0.6$, the LLN approximate optimal value is $r=1/3$ ($k^*=20$), whereas the exact value is $k^*=15$. Third, the minimum $\expec[Y_{k:n}]$'s in the LLN are close to the values in the exact result. Therefore, we see that in spite of some differences, the LLN approximation is good, as it shows well the general trend and gives some  exact results. Furthermore, the approximation should be even better for larger $n$.

\subsection{Data-Dependent Scaling}
Under the data-dependent scaling, the service time of a task consisting of $s$ CUs is given by $Y=s\cdot\Delta+X$. Therefore, the job completion time is given by
\[
Y_{k:n}=s\cdot\Delta + X_{k:n}, ~~\text{where}~~X \sim \BMD(B, \epsilon)
\]
and, by using the expression for $\expec[X_{k:n}]$ in  \eqref{Eq:bimodal_server}, we have
\begin{equation}
\label{Eq:bimodal_data}
    \expec[Y_{k:n}]= s\Delta+1+(B-1)\sum^{k-1}_{i=0} \binom{n}{i} \epsilon^{n-i}(1-\epsilon)^{i}
\end{equation}

Note that $s\Delta$ is a decreasing function of $k$ (since $s=n/k$), whereas $\expec[X_{k:n}]$ is an increasing function of $k$ (by the definition of order statistics). Then, there is a balance between $s\Delta$ and $\expec[X_{k:n}]$ that minimizes the expected job completion time $\expec[Y_{k:n}]$. However, since the expression of $\expec[Y_{k:n}]$ is very complicated, it is difficult to find the minimum value.
Instead of finding the exact value of minimum $\expec[Y_{k:n}]$, we can find the approximation by applying law of large numbers for the large $n$ scenario.

{\it Large n Scenario:}

By applying LLN, we find the approximation for the expected job completion time $\expec[Y_{k:n}]$ in Theorem~\ref{Thm:bimodal_LLN2}.
\begin{theorem}
\label{Thm:bimodal_LLN2}
Considering $\BMD(B, \epsilon)$ service time with data-dependent scaling, when $n$ is sufficiently large, we find the LLN approximation for $\expec[{Y_{k:n}}]$,
\begin{equation}
\label{eq:BMD-LLN2}
    \expec[{Y_{k:n}}] \sim\frac{\Delta}{r}+ p_r + B q_r,~~\text{as}~~n\rightarrow \infty
\end{equation}
where $r=\frac{k}{n}$ is the code rate, $p_r\rightarrow 1$ if $ 1 - \epsilon > r$ and 0 if the inequality is reversed, and $q_r=1-p_r$.
\end{theorem}

\begin{proof}
At server $i$, $i=1,\dots, n$, the task completion time $Y_i=s\Delta+X_i$, where $X_i$ is a Bi-Modal random variable taking value $1$ or $B$. Therefore, the job completion time $Y_{k:n}=s\cdot\Delta + X_{k:n}$. We define an indicator function $\mathbf{1}_{\{X_i|\{1\}\}}:\{X_i|\{1,B\}\}\rightarrow \{0,1\}$. Let $M$ be the sum of $n$ i.i.d. indicators $\mathbf{1}_{\{X_i|\{1\}\}}$ where $i\in\{1, \cdots, n\}$. Then $
    {\Pr}\{X_{k:n}=B\} = {\Pr}\{M<k\}~~ \text{and} ~~  {\Pr}\{X_{k:n}= 1\} = {\Pr}\{M\ge k\}$ 

We next look into $M$ through the LLN lens. Let $r\doteq\frac{k}{n}$ and $M_n\doteq M/n$. Observe that $M_n\rightarrow1-\epsilon$ as $n \rightarrow \infty$. From the above, we see that $p_r= {\Pr} \{ X_{k:n} = 1\}={\Pr}\{M_n\ge r\}\rightarrow 1$ if $ 1 - \epsilon > r$ and 0 if the inequality is reversed, and $q_r=1-p_r$. Since $Y_{k:n}=s\cdot\Delta + X_{k:n}$ and $s=1/r$, the LLN approximation yields the expression \eqref{eq:BMD-LLN2}.
\end{proof}

From \eqref{eq:BMD-LLN2}, we see that $\expec[{Y_{k:n} }]$ is a convex, unimodal function of $r$ on $[0,1-\epsilon]$ and a decreasing function of $r$ on $(1-\epsilon,1]$. Therefore, $\expec[{Y_{k:n} }]$ has two local minimums: $1+ \Delta/(1 - \epsilon)$ at $r=1-\epsilon$ and $\Delta + B$ at $r=1$, which we compare and reach the following conclusions. When 
$\epsilon \le (B-1)/(\Delta + B -1)$, the global minimum is $1+ \Delta/(1 - \epsilon)$, and thus  coding at rate rate $r=1-\epsilon$ is optimal. When $\epsilon > (B-1)/(\Delta + B -1)$, the global minimum is $\Delta + B$, and thus splitting is optimal.

\paragraph*{Numerical Analysis}
In Fig.~\ref{fig:data_Bi-Modal_e}, we evaluate  $\expec[{Y_{k:n} }]$ vs. $k$. There are $n=12$ workers and $\Delta=5$. By comparing five different values of $\epsilon\in\{0.05,0.2,0.5,0.6,0.9\}$, we observe that when $\epsilon\rightarrow0$ (e.g. $0.05$), $\expec[{Y_{k:n} }]$ decreases with increasing $k$, thus splitting is optimal. When $\epsilon$ is small (e.g. $0.2$, $0.5$), $\expec[{Y_{k:n} }]$ reaches its minimum at $k=6$, thus coding is optimal. When $\epsilon$ is large (e.g. $0.6$, $0.9$), $\expec[{Y_{k:n} }]$ decreases with increasing $k$ again, thus splitting is optimal. 
\begin{figure}
    \centering
    \includegraphics[width=0.4\textwidth]{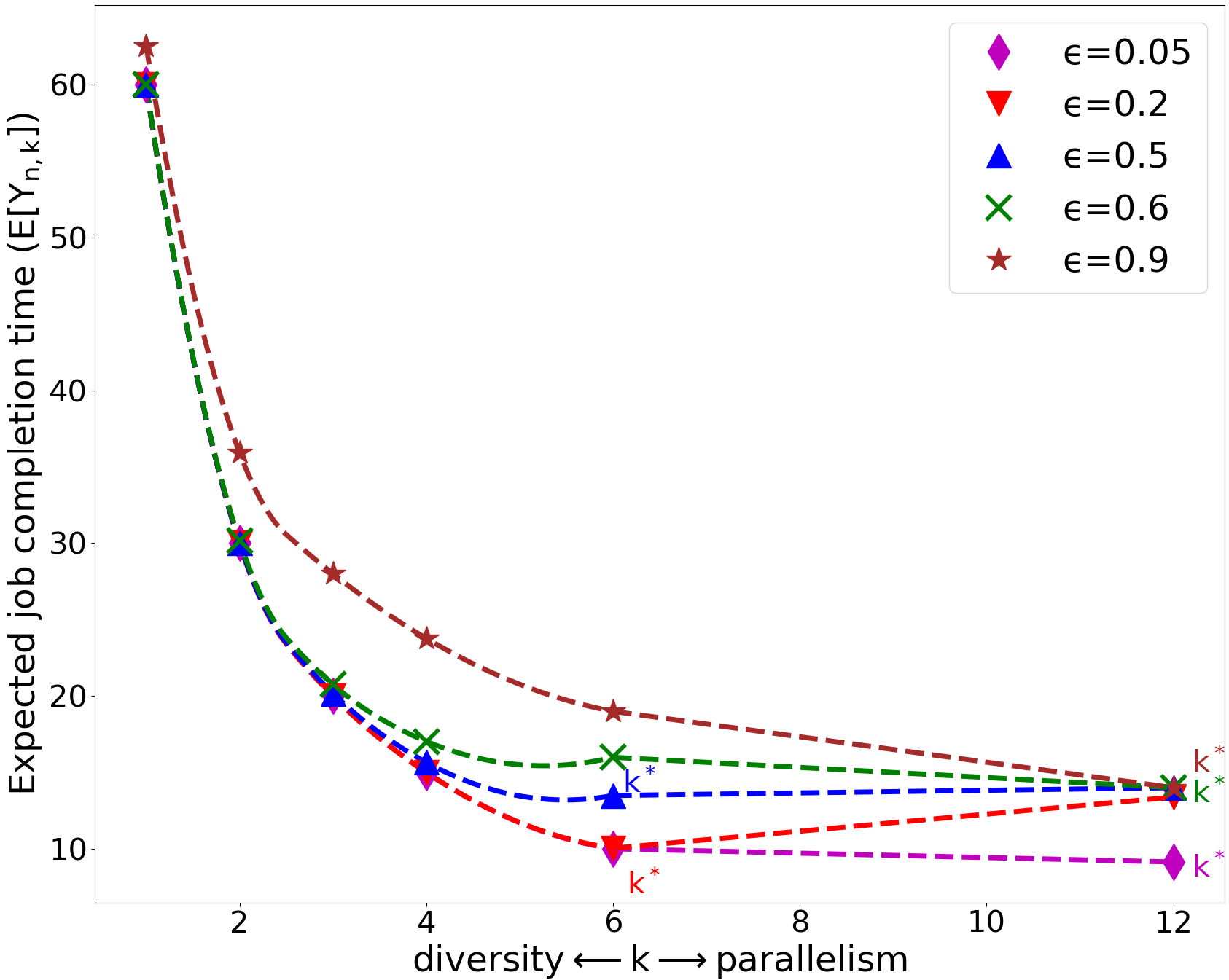}
    \caption{Expected job completion time $\expec[Y_{k:n}]$ for Bi-Modal service time with data-dependent scaling as a function of the diversity/parallelism parameter $k$ (cf.~\eqref{Eq:bimodal_data}). The straggling magnitude $B$ is $10$. The number of workers (job size) is $n=12$, and task size per worker is $s=n/k$ (Since both $k$ and $s$ are integers, we have $k\in\{1,2,3,4,6,12\}$. We use dashed curves to connect the points corresponding to different allowed values of $k$ for a given $\alpha$.)  When $\epsilon$ is small, the optimal code rate decreases with increasing $\epsilon$; When $\epsilon$ is relatively large, maximal parallelism is optimal.}
    \label{fig:data_Bi-Modal_e}
\end{figure}
From these observations, we conclude that as $\epsilon$ increases, $\expec[{Y_{k:n} }]$ is minimized by introducing more diversity, but when $\epsilon$ approaches $1$, maximal parallelism is optimal.

\begin{figure}
    \centering
    \includegraphics[width=0.4\textwidth]{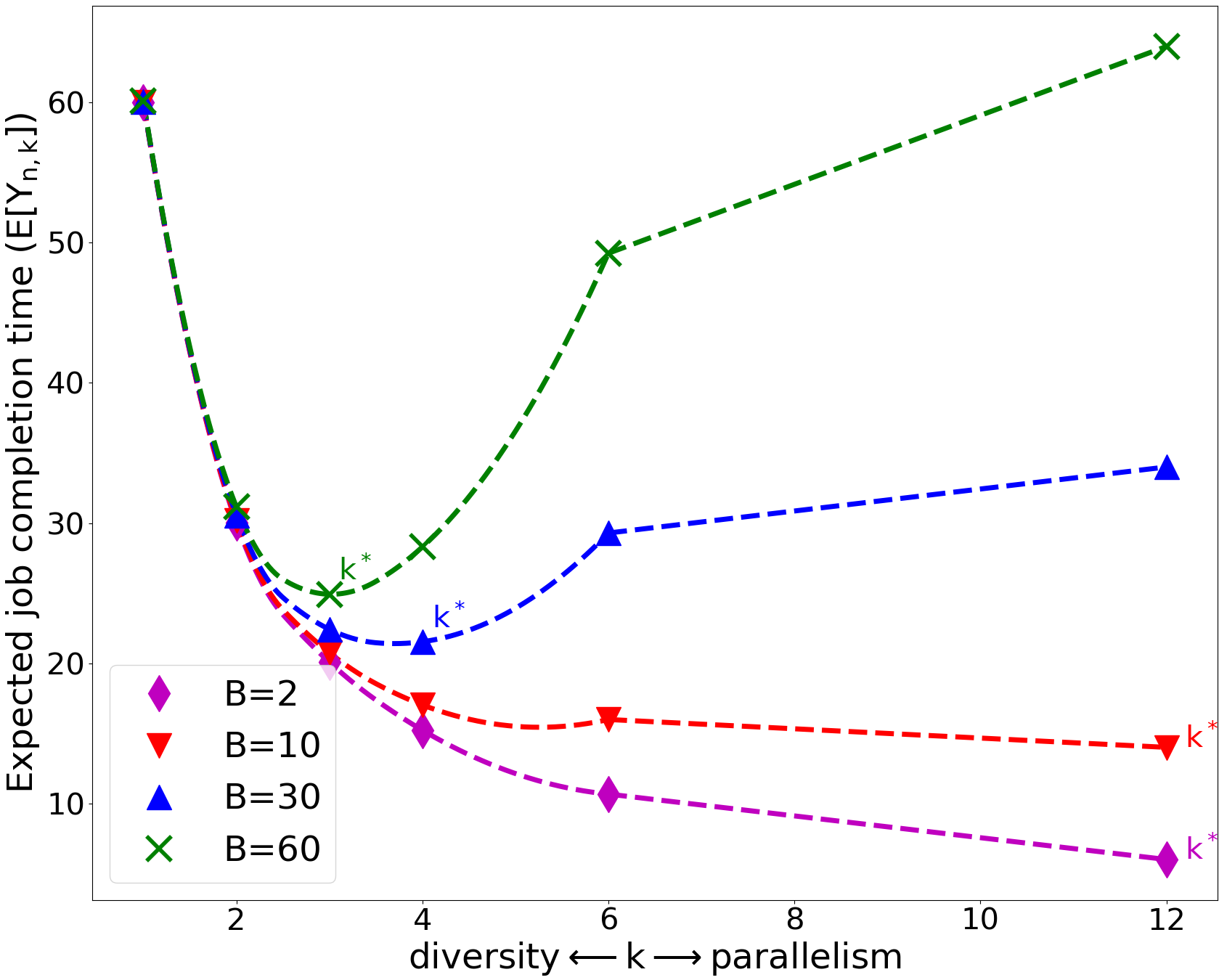}
    \caption{Expected job completion time $\expec[Y_{k:n}]$ for Bi-Modal service time with data-dependent scaling as a function of the diversity/parallelism parameter $k$ (cf.~\eqref{Eq:bimodal_data}). The straggling probability $\epsilon$ is $0.6$. The number of workers (job size) is $n=12$, and task size per worker is $s=n/k$ (Since both $k$ and $s$ are integers, we have $k\in\{1,2,3,4,6,12\}$. We use dashed curves to connect the points corresponding to different allowed values of $k$ for a given $\alpha$.) The optimal code rate decreases with increasing $B$.}
    \label{fig:data_Bi-Modal_B}
\end{figure}
In Fig.~\ref{fig:data_Bi-Modal_B}, we analyze $\expec[{Y_{k:n} }]$ vs. $k$ in a system with $n=12$ workers and $\Delta=5$ for four different values of $B\in\{2,10,30,60\}$. The diversity/parallelism trade-off is determined by the magnitude of $B$. When $B$ is small (e.g. $2$, $10$), $\expec[{Y_{k:n} }]$ decreases with increasing $k$, thus splitting is optimal. When $B$ is large (e.g. $30$, $60$), $\expec[{Y_{k:n} }]$ reaches its minimum at $k=6$, thus coding is optimal. Notice that if $\Delta=0$, replication is optimal; If $\Delta\gg B$, splitting is optimal.

\begin{figure}[hbt]
	\centering
	\includegraphics[width=0.4\textwidth]{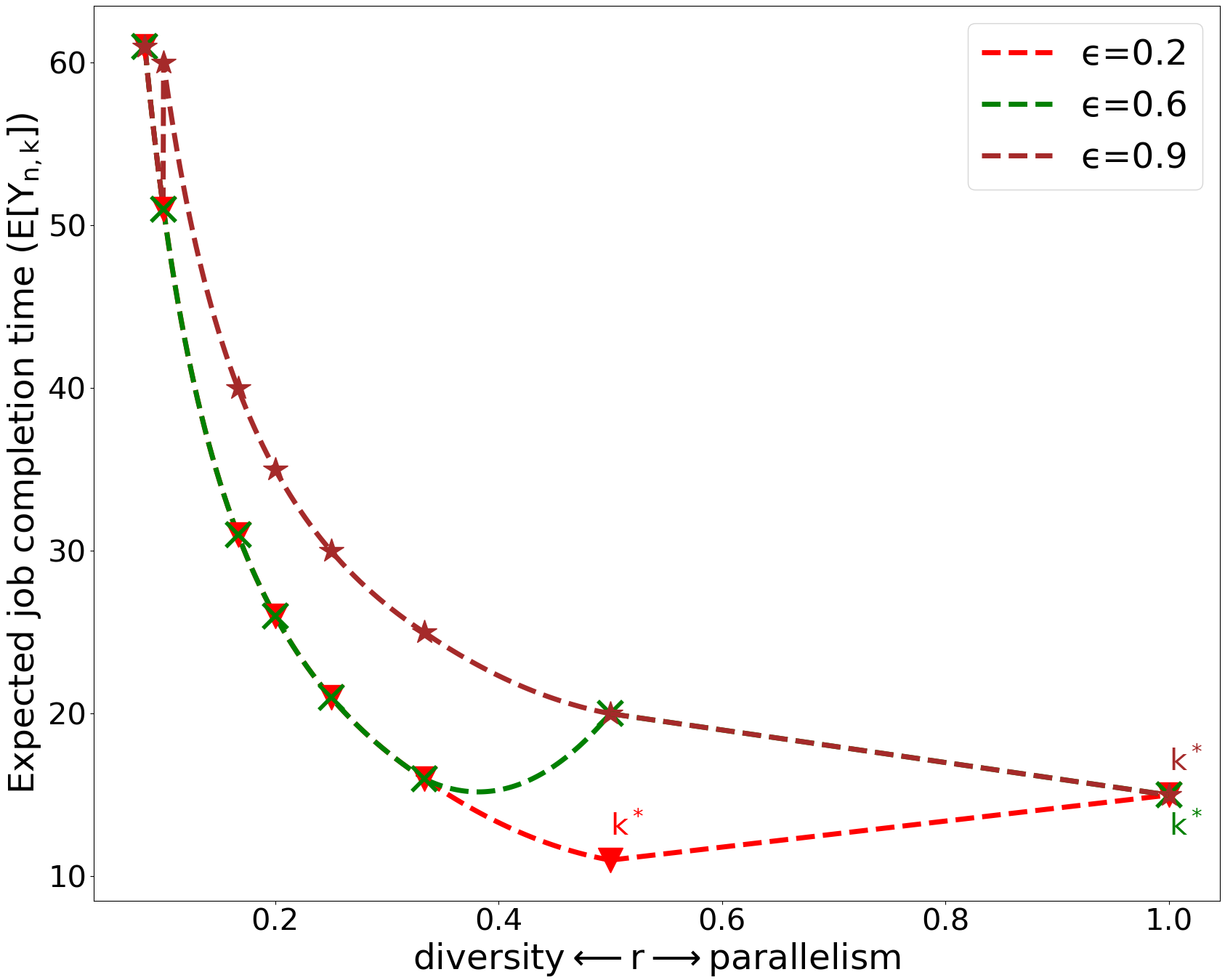}~\\
	\includegraphics[width=0.4\textwidth]{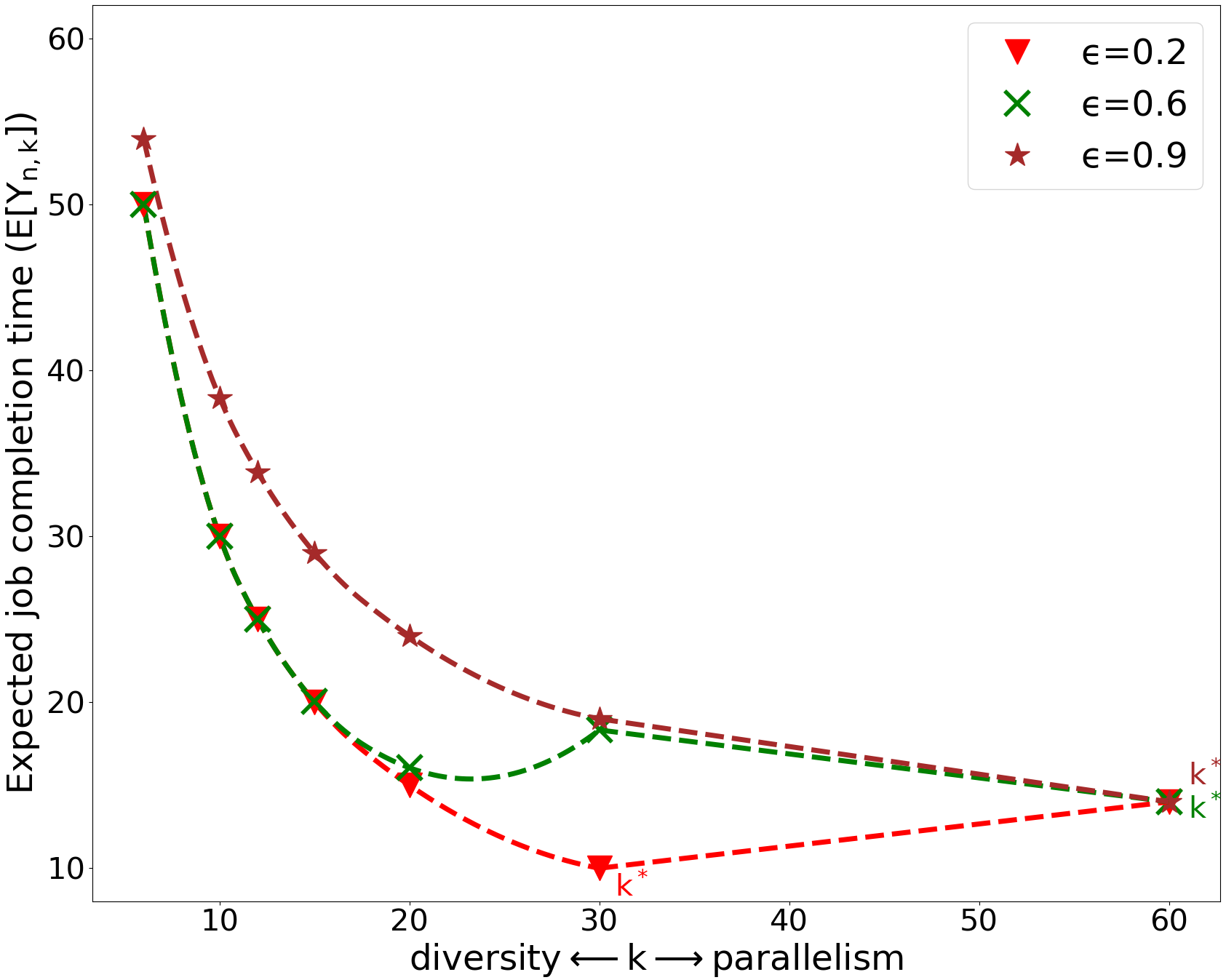}~
	\caption{Comparisons between the LLN approximation of $\expec[Y_{k:n}]$ given by \eqref{eq:BMD-LLN2} and the exact result given by \eqref{Eq:bimodal_data} for Bi-Modal service time with data-dependent scaling. The straggling magnitude $B$ is $10$. The number of workers (job size) is $n=60$, task size per worker is $s=n/k$. (upper) The LLN approximation shows the $\expec[Y_{k:n}]$ vs.\ the code rate $r=k/n$. (lower) The exact dependence of $\expec[Y_{k:n}]$ on the diversity/parallelism parameter $k$. (Since both $k$ and $s$ are integers, we have $k\in\{1,2,3,4,5,6,10,12,15,20,30,60\}$. We use dashed curves to connect the points corresponding to different allowed values of $k$ for a given $\alpha$.) The LLN approximation performs well here.}
	\label{fig:data_Bi-Modal_LLN}
\end{figure}
In Fig.~\ref{fig:data_Bi-Modal_LLN}, we compare the LLN approximation of $\expec[Y_{k:n}]$ (upper) with the exact result  \eqref{Eq:bimodal_data} (lower). There are $n=60$ workers, $\Delta=5$, and three values of $\epsilon$. Since $\expec[Y_{k:n}]$ is very large when $k$ ($r$) is small, we only plot the points for $k\ge 5$ ($r>1/12$). We compare three important metrics to evaluate the approximation: the local minimums, the optimal $k^*$ and the minimum $\expec[Y_{k:n}]$, and observe the following. First, for each value of epsilon, both the LLN and the exact result have the same number of local minimums. The values of local minimums in both graphs are close to each other. Second, the LLN shows the same values of $k^*$'s as the exact result. Third, the minimum $\expec[Y_{k:n}]$'s obtained by the LLN approximation are close to the exact values. Overall, the LLN gives a very good approximation to the exact result, and the approximation will be more even more accurate when $n$ is larger.

\subsection{Additive Scaling}
Under the additive scaling, the service time of a task consisting of $s$ CUs is given by 
\[
Y=X_1+\dots +X_{s}, ~~ \text{where}~ X_i\sim \BMD(B,\epsilon).
\]
We derive the expressions for $Y$ and the expected job completion time $\expec[Y_{k:n}]$ in Lemma~\ref{Le:biorder} (see Appendix~\ref{AP:bimodal}). These expressions are unsuitable for theoretical analysis, but can be numerically evaluated. By theoretical analysis, we only find that splitting is optimal when $B\le 2$ in Proposition~\ref{Th:Bi-Modal-add}.

\begin{proposition}
\label{Th:Bi-Modal-add}
 For $\BMD(B,\epsilon)$ service time, if $B \le 2$, the expected job completion time $\expec[Y_{k:n}]$ reaches its minimum when $k=n$ (maximal parallelism).
\end{proposition}
\begin{proof}
When $k=n$, we have $s=n/k=1$, and thus $Y\sim \BMD(B,\epsilon)$. Then we have $Y_{n:n}=1$ with the probability $(1-\epsilon)^n$, and $Y_{n:n}=B$ with the probability $1-(1-\epsilon)^n$. Therefore, $\expec[Y_{n:n}]\le B\le2$. When $k<n$, we have $s=n/k\ge2$, then $\expec[Y_{k:n}]>2$.
\end{proof}

{\it Numerical Analysis:}
In Fig.~\ref{fig:add_Bi-Modal_e}, we evaluate $\expec[Y_{k:n}]$ vs.\ $k$. We consider a system with $n=12$ workers for four different values of $\epsilon\in\{0.005,0.2,0.6,0.9\}$. Some observations can be made from the figure:
when $\epsilon\rightarrow 0$ (e.g. $0.005$), splitting outperforms the other two strategies slightly.
 When $\epsilon$ is small (e.g. $0.2$), there is a balance between diversity and parallelism, and coding with the code rate $1/2$ is optimal.
When $\epsilon$ is large (e.g. $0.6$, $0.9$),  splitting is optimal.
\begin{figure}
    \centering
    \includegraphics[width=0.4\textwidth]{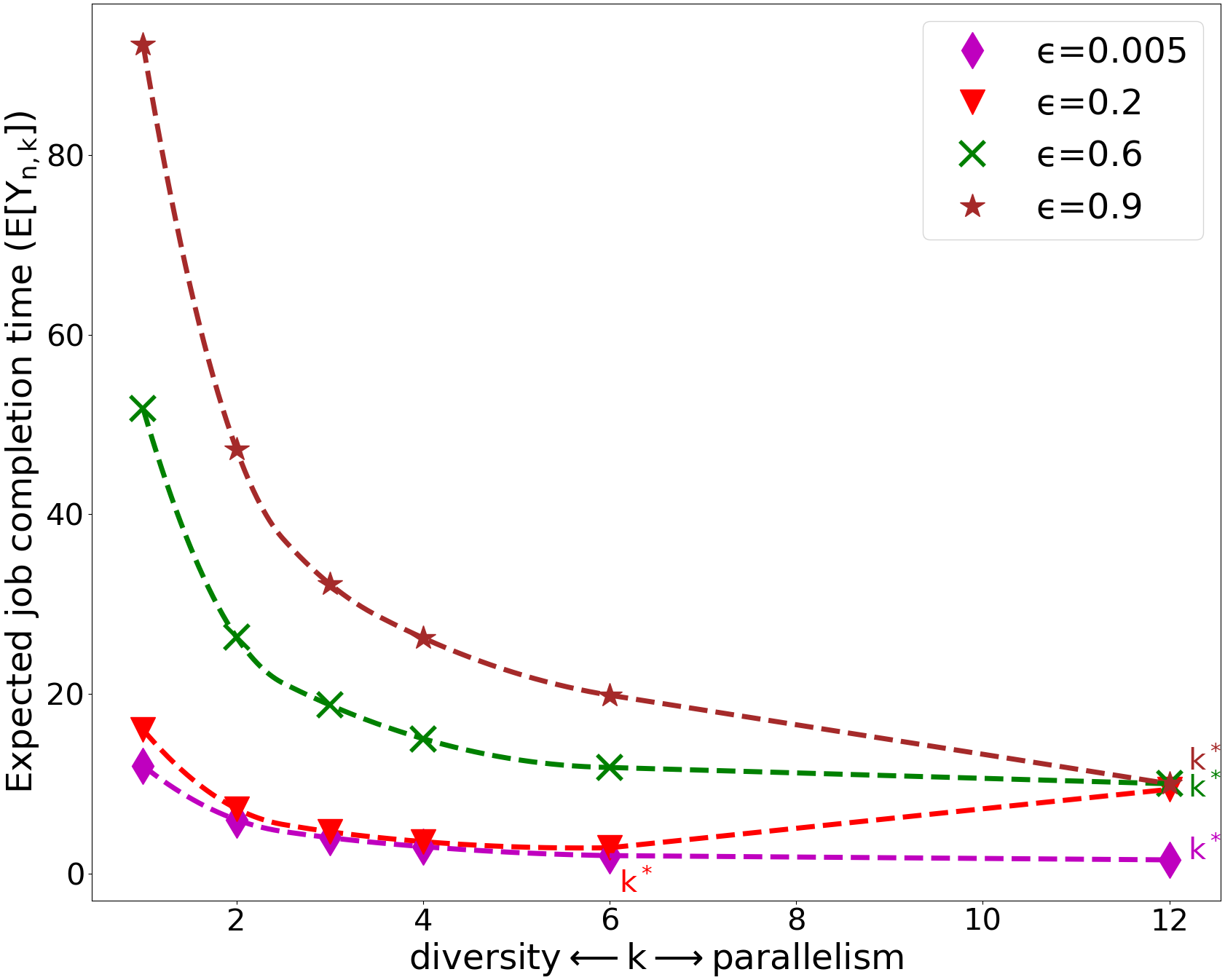}
    \caption{Expected job completion time $\expec[Y_{k:n}]$ for Bi-Modal service time with additive scaling as a function of the diversity/parallelism parameter $k$ (cf.~\eqref{Eq:BiModal}). The straggling magnitude $B$ is $10$. The number of workers (job size) is $n=12$, and task size per worker is $s=n/k$ (Since both $k$ and $s$ are integers, we have $k\in\{1,2,3,4,6,12\}$. We use dashed curves to connect the points corresponding to different allowed values of $k$ for a given $\alpha$.) When $\epsilon$ approaches to $0$ or $1$, maximal parallelism is optimal. Otherwise, coding is optimal, and the code rate is around $1/2$. }
    \label{fig:add_Bi-Modal_e}
\end{figure}
These observations are similar to those for server-dependent and data-dependent scaling. 
We conjecture that coding with a proper code rate is always better than replication in Conjecture~\ref{CJ:BiModal}.  Recall that this is not the case for server-dependent and data-dependent scaling, where replication may be optimal for certain (large) values of $B$.
\begin{conjecture}
\label{CJ:BiModal}  
For $\BMD(B,\epsilon)$ service time with additive scaling, either coding or splitting outperforms replication, i.e. there exists $k$, $2\le k\le n$ such that $\expec[Y_{k:n}]<\expec[Y_{1:n}]$.
\end{conjecture}

In Fig.~\ref{fig:add_Bi-Modal_B}, we plot $\expec[Y_{k:n}]$ vs.\ $k$ in a system with $n=12$ workers for four different values of $B\in\{2,5,10,20\}$. We see that the diversity/parallelism trade-off is determined by the magnitude of straggling $B$. The figure also shows that the optimal code rate is either $1/2$ or $1$, which coincides with Conjecture~\ref{CJ:BiModal}. To examine this result, we evaluated $\expec[Y_{k:n}]$ for values of $B$ from $2$ to $10000$. We observed that when $B\le 106$, the optimal code rate is $1/2$, and when $B\ge 107$, the optimal code rate is $1/3$. 
These simulation results provide some support to Conjecture~\ref{CJ:BiModal}. 
\begin{figure}
    \centering
    \includegraphics[width=0.4\textwidth]{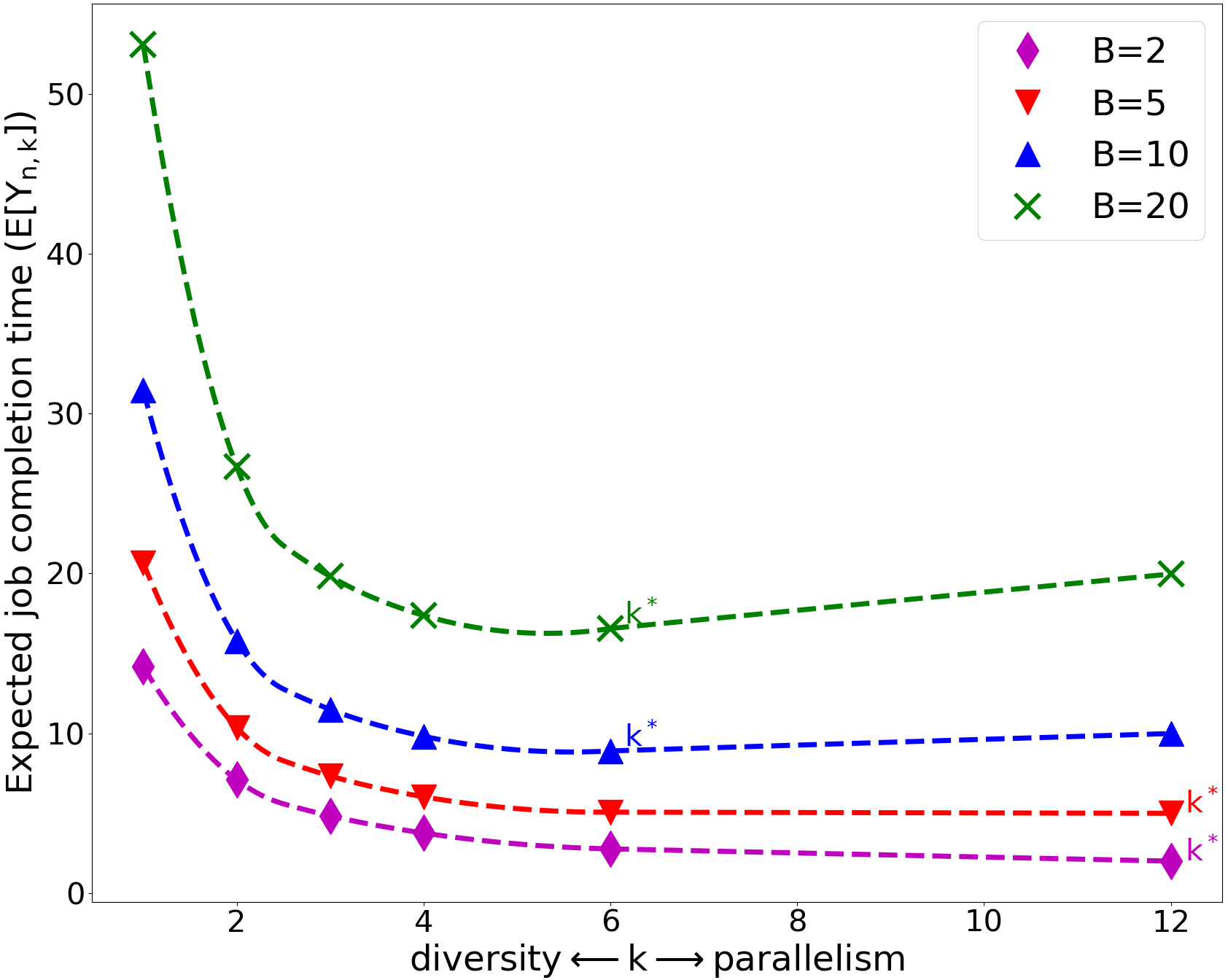}
    \caption{Expected job completion time $\expec[Y_{k:n}]$ for Bi-Modal service time with additive scaling as a function of the diversity/parallelism parameter $k$ (cf.~\eqref{Eq:BiModal}). The straggling probability $\epsilon$ is $0.4$. The number of workers (job size) is $n=12$, and task size per worker is $s=n/k$ (Since both $k$ and $s$ are integers, we have $k\in\{1,2,3,4,6,12\}$. We use dashed curves to connect the points corresponding to different allowed values of $k$ for a given $\alpha$.) The optimal code rate decreases as $B$ increasing, reaching the minimum at $1/3$.}
    \label{fig:add_Bi-Modal_B}
\end{figure}

\paragraph*{Remark}
Based on the insight we gained from the analysis up to this point, we make the following conjecture about the optimal strategy for general service time PDFs. The claim in the conjecture holds for the PDFs considered in this paper.
\begin{conjecture}
\label{conject:additive}
Under additive scaling and a general CU service time, either coding or splitting outperforms replication.
\end{conjecture}

\section{Conclusions and Future Directions}
\label{sec:conclusions}
In distributed computing with redundancy, smaller tasks that make up a large computing job are executed in parallel, and redundancy is added as a form of diversification that reduces the job service time dependence on the execution of straggling tasks. Both parallelism and diversity reduce job completion time. However, in systems where a constant number of workers is available for job execution, more redundancy means less parallelism and vice versa. We considered the
 trade-off between diversity and parallelism with the purpose to minimize the job completion time. 
 
 Depending on the level of redundancy used in the system, each worker has to execute a task consisting of one or more computing units (a minimum-size task below which distributed computing would be inefficient). In Sections \ref{sec:shift}, \ref{sec:Pareto}, and \ref{sec:Bi-Modal}, we considered three common models for the computing unit service time PDF. For each of these models, we adopted three common assumptions about service time scaling with the task size. For each service time model, the results are summarised in a table at the end of the  corresponding section. We also drew several conclusions and conjectures about a general service time distributions with the finite fourth moment.

In our summary of the results in Table~\ref{Table:models}, we distinguished only between the following three regimes: 1) maximum parallelism (splitting the job across the workers),  2) maximum diversity (replication of the job at each worker), and 3) the region where coding is used to enable a trade-off between diversity and parallelism. We indicated in the table how the optimal strategy changes as the tail of the service time PDF becomes heavier. 
The general conclusion is that the optimal level of redundancy strongly depends on the assumptions made about the task service time PDF and its scaling with the task size. This work sets the stage for many problems of interest to be studied in the future. We briefly describe three directions of immediate interest.

\subsubsection{Diversity/parallelism tradeoff for non-MDS codes} As we mentioned in Sec.~\ref{sec:coding}, our system model and analysis approach are not limited to MDS codes. 
If an $[n, k]$ code with minimum distance $d$ is used, then the job is completed when any $m = n - (d - 1)$ out of $n$ tasks are completed. The Singleton bound imposes the constraint $m \ge k$, where $m = k$ for MDS codes.
When $m>k$, the diversity/parallelism tradeoff may change, since for such systems, the task size $s=n/k$ is the same as the MDS coded ones, but they are able to mitigate fewer stragglers. 
Of particular interest is characterizing the diversity/parallelism tradeoff for codes with features that are attractive in practice, e.g., the codes proposed in \cite{yu2020straggler} with linear encoding and decoding. Observe that characterizing the diversity/parallelism tradeoff is a complementary task to code design and selection. It determines the code parameters within a class of codes that are optimal under the given system model.

\subsubsection{Diversity/parallelism tradeoff for for general job sizes}
Since we are concerned with systems with a fixed number of workers $n$, we have assumed that each job is split into $n$ CUs (has size $n$). To generalize this assumption, we can set the job size to be $bn$ CUs, where $b$ is an integer and $b\ge 2$. This generalization will require that some of our results be slightly modified (Theorem.~\ref{thm:server}, \ref{Th:Paretomethod} and \ref{Thm:bimodal_LLN}).  Some other claims may need new statements or proofs, e.g., Theorem.~\ref{th:splitvsrepli} and \ref{LLN:Pareto}. This generalization will allow us to analyze the full spectrum of code rates. Further generalizations of job sizes may be much more difficult but worth studying.

\subsubsection{Diversity/parallelism tradeoff for other system models} 
This paper analyzed the most common service time and scaling models. Some other distributions, e.g., Weibull distribution used in \cite{reisizadeh2019coded} are also of practical interest. In previous sections, we stated Conjectures \ref{conject:server}, \ref{CJ:BiModal} and \ref{conject:additive} regarding server-dependent and additive scaling. Conjectures \ref{conject:server} and \ref{conject:additive} concern general service time distributions, but knowing whether they hold for classes of distributions other than those considered in the paper is of interest.
This paper considered distributed, parallel  architectures commonly implemented in modern computing frameworks, e.g., Kubernetes and Apache Mesos, and adopted the model that corresponds to these systems. However, other computing architectures which give rise to different models are also important to study. One such example is the emerging wireless edge computing where, for example, the worker nodes may not be statistically identical.


\appendices
\section{Mathematical Background}
\subsection{Order Statistics}
In executing jobs with redundant tasks, the notion of order statistics plays a central role. We here state the results we use throughout the paper.
More information can be found in, e.g., \cite{Order:Renyi53,Order:ArnoldBN08,Order:David11, Pareto:Arnold15}.
Let $X_1,X_2, \dots, X_n$ be $n$ samples of some RV $X$. Then the $k$-th smallest is an RV denoted by $X_{k:n}$ and known as the $k$-th order statistic of $X_{1},X_{2},\cdots,X_{n}$ . 
%

\subsubsection{Exponential Distribution} If $X_{1},X_{2},\cdots,X_{n}$ are $\expD(W)$, then $X_{1:n}$ is $\expD(W/n)$, and
\begin{equation}
X_{k:n}=X_{1:n}+X_{1:(n-1)}+\dots +X_{1:(n-k+1)}. 
\label{eq:OSsum}
\end{equation}
The expectation of $X_{k:n}$ is given by
\begin{align}
\label{eq.orderstat_mean_variance}
& \expec[X_{k:n}] = W\sum_{i=1}^{k} {\frac{1}{n-k+i}} = W \bigl(H_n - H_{n-k}\bigr)
\end{align}
where $H_n$ is (generalized) harmonic numbers defined as
$H_n =  \sum_{j=1}^n \frac{1}{j}$.  We  often use the approximation
$H_n=\log n+\gamma + \mathcal{O}(n^{-1})$, where
$\gamma = 0.577$ is Euler's constant.


\subsubsection{Erlang Distribution}
If $X_{1},X_{2},\cdots,X_{n}$  are $\erD(s,W)$, then, according to the formula of gamma order statistics in \cite{gupta1960order}, we have
\begin{equation}
\begin{aligned}
        X_{k:n} = & \frac{Wk}{(s-1)!} \binom{n}{k} \sum^{k-1}_{i=0} (-1)^{i} \binom{k-1}{i}\\& \sum^{(s-1)(n-k+i)}_{j=0}\!\!\! \alpha_{j}(s,n-k+i) \frac{(s+j)!}{(n-k+i+1)^{s+j+1}}
        \label{Eq:Erlang-order}
\end{aligned}
\end{equation}
where $\alpha_{z}(x,y)$ is the coefficient of $t^{z}$ in the expansion of $\bigl(\sum^{x-1}_{l=0} t^{l}/l!\bigr)^{y}$. 

\subsubsection{Pareto Distribution}
If $X_{1},X_{2},\cdots,X_{n}$ are $\parD(\lambda, \alpha)$, then the expectation of $X_{k:n}$ for $\alpha > 1$ is given by
\begin{equation}
\label{Eq:Pareto}
    \expec[X_{k:n}] = \lambda\frac{n!}{(n-k)!}\frac{\Gamma(n-k+1-1/\alpha)}{\Gamma(n+1-1/\alpha)} 
\end{equation} 
where the complete gamma function is defined as $\Gamma(x) = \int_0^{\infty} u^{x-1}e^{-u}du$. 
%

In our work on straggler mitigation (see \cite{aktas2017effective}), we have obtained the following approximation: 
\begin{equation}
\label{Ap:Pareto}
    \Gamma(x + \beta)/\Gamma(x + \alpha) \sim x^{\beta-\alpha}
\end{equation} 
by using Stirling's approximation or by an induction on Gautschi's inequality \cite{Ineq:Gautschi59}. This result is useful in finding numerically good approximations of $\expec[X_{k:n}]$ for large $n$.

\subsubsection{Bi-Modal Distribution}
\label{AP:bimodal}

\begin{lemma}
\label{Le:biorder}
If $Y=X_1+\cdots +X_s$ is the sum of $s$ i.i.d. $\BMD(B,\epsilon)$ RVs, then
\begin{equation}
Y=s-w+wB ~~\text{w.p.}~~\textstyle{\binom{s}{w}
(1-\epsilon)^{s-w}\epsilon^{w},} ~~0\le w\le s.
\label{eq:BMDa}
\end{equation}
The expectation of $Y_{k:n}$ is
\begin{equation}
\label{Eq:BiModal}
    \begin{aligned}
    \expec[Y_{k:n}]&=s+(B-1)\sum_{w=1}^{s-1}w\sum^{k-1}_{i=0} \binom{n}{i}\bigr(\sum^{w-1}_{j=0} p_{j}\bigl)^{i} \\&\bigl[\sum^{n-i}_{l=k-i} \binom{n-i}{l}  p_{w}^{l}
\bigl(\sum^{s}_{h=w+1}p_{h}\bigr)^{n-i-l}\bigr]\\&+s(B-1)\sum^{k-1}_{i=0} \binom{n}{i} p_{s}^{n-i}(1-p_{s})^{i}
\end{aligned}
\end{equation}

\end{lemma}
\noindent
\begin{proof}
The expression in \eqref{eq:BMDa} is straightforward to derive.
Since $Y_{1},\cdots,Y_{n}$ are sums of $s$ i.i.d. $\BMD(B,\epsilon)$ RVs, we deduce $Y_{k:n}$ by using the definition of order statistics,
\begin{align*}
Y_{k:n} = \begin{cases}
s ~~&\text{w.p.} ~~ \textstyle{\sum^{n-k}_{i=0} \binom{n}{i} p_{0}^{n-i}(1-p_{0})^{i}}  \\
s-w+wB~~&\text{w.p.} ~~\textstyle{\Pr(w)}, 1\le w<s\\
sB ~~&\text{w.p.} ~~ \textstyle{\sum^{k-1}_{i=0} \binom{n}{i} p_{s}^{n-i}(1-p_{s})^{i}}
\end{cases} 
\end{align*}
where $p_{v}=\binom{s}{v}
(1-\epsilon)^{s-v}\epsilon^{v} ~~\text{for}~~ v=0,\dots, s$. When $w=0$ and $w=s$, the respective probabilities of $Y_{k:n} = s$ and $Y_{k:n} = sB$ are straightforward to derive. We consider the case when $0<w<s$ and $Y_{k:n} = s-w+wB$. Since we are concerned with $k$-th order statistics, we know that there are at most $k-1$ RVs among $\{Y_{1},\cdots,Y_{n}\}$ whose values are smaller than $s-w+wB$. Consider the event that $i$ ($i\le k-1$) of the RVs are smaller than $s-w+wB$. The probability of this event is $\binom{n}{i}\bigr(\sum^{w-1}_{j=0} p_{j}\bigl)^{i}$.  Among the remaining $n-i$ RVs in $\{Y_{1},\cdots,Y_{n}\}$, there are at most $n-k$ RVs whose values are larger than $s-w+wB$; the others are equal to $s-w+wB$. Consider an event where $\ell$ ($k-i\le \ell \le n-i$) of these RVs take values larger than $s-w+wB$. Thus all other $n-i-\ell$ RVs take value $s-w+wB$. The probability of this event under the condition of the previous event is
$\binom{n-i}{l} p_{w}^{l}\bigl(\sum^{s}_{h=w+1}p_{h}\bigr)^{n-i-l}$. 
Thus,
\begin{align*}
 \Pr(w)=&\sum^{k-1}_{i=0} \binom{n}{i}\bigr(\sum^{w-1}_{j=0} p_{j}\bigl)^{i} \\&\bigl[\sum^{n-i}_{l=k-i} \binom{n-i}{l} p_{w}^{l}
\bigl(\sum^{s}_{h=w+1}p_{h}\bigr)^{n-i-l}\bigr].
\end{align*}

Therefore, we have
\begin{align*}
    \expec[Y_{k:n}]&=s+(B-1)\sum_{w=1}^{s-1}w\sum^{k-1}_{i=0} \binom{n}{i}\bigr(\sum^{w-1}_{j=0} p_{j}\bigl)^{i} \\&\bigl[\sum^{n-i}_{l=k-i} \binom{n-i}{l}  p_{w}^{l}
\bigl(\sum^{s}_{h=w+1}p_{h}\bigr)^{n-i-l}\bigr]\\&+s(B-1)\sum^{k-1}_{i=0} \binom{n}{i} p_{s}^{n-i}(1-p_{s})^{i}
\end{align*}
\end{proof}

\subsection{A Generalized Birthday Problem}
\label{Ap:birthday}
A generalized birthday problem is stated as follows: ``How many draws with replacement on average have to be made from a set of $n$ coupons until one of the coupons is drawn $d$ times?'' The expected number of draws $\expec(n,d)$ was determined in \cite{klamkin1967extensions}:
\begin{equation}
\label{Eq:birth}
\begin{aligned}
    &\expec(n,d)=\int_0^\infty \!\! e^{-t} \Bigl[S_d\Bigl(\frac{t}{n}\Bigr)\Bigr]^ndt~~~ \\&\text{where}~~
    S_d(x)=1+\frac{x}{1!}+\frac{x^2}{2!}+\dots+\frac{x^{d-1}}{(d-1)!}.
    \end{aligned}
\end{equation}
We also use an asymptotic expression for (\ref{Eq:birth}) given in 
\cite{klamkin1967extensions}, which says that for a fixed $d$, we have
\begin{equation}
\label{Eq:birthapprox}
    \expec(n,d) \sim \sqrt[d]{d!}\,\Gamma(1+1/d)n^{1-\frac{1}{d}}, ~~\text{as}~~n\rightarrow \infty.
\end{equation}

\bibliographystyle{IEEEtran}
\bibliography{DistRedundan,stragglers,CCML,BallsUrns,ReferencesOrganized}
\begin{IEEEbiographynophoto}{Pei Peng} 
received his B.S. degree in information engineering from South China University of Technology, Guangzhou, China, in 2011, and the M.S. degree in electronics and communication engineering from Shanghai Institute of Micro-system and Information Technology, Chinese Academy of Science, Shanghai, China, in 2014. 
Since 2015, he has been a Ph.D.\ student in the electrical and computer engineering department at Rutgers, the State University of New Jersey, USA. During his PhD.\ studies, Pei Peng has served as both research and teaching assistant and has received a 2020 ECE department teaching award. His research interests are coding and allocation in distributed computing systems, covert communications, and machine learning.
\end{IEEEbiographynophoto}

\begin{IEEEbiographynophoto}{Emina Soljanin}
is a professor at Rutgers University. Before moving to Rutgers in 2016, she was a (Distinguished) Member of Technical Staff for 21 years in the Mathematical Sciences Research of Bell Labs.
Her interests and expertise are wide. Over the past quarter of the century, she has participated in numerous research and business projects, as diverse as power system optimization, magnetic recording, color space quantization, hybrid ARQ, network coding, data and network security, distributed systems performance analysis, and quantum information theory. She served as an Associate Editor for Coding Techniques, for the IEEE Transactions on Information Theory, on the Information Theory Society Board of Governors, and in various roles on other journal editorial boards and conference program committees.
 Prof.~Soljanin an IEEE Fellow, an outstanding alumnus of the Texas A\&M School of Engineering, the 2011 Padovani Lecturer, a 2016/17 Distinguished Lecturer, and 2019 President of the IEEE Information Theory Society.
 \end{IEEEbiographynophoto}
 
\begin{IEEEbiographynophoto}{Philip Whiting}
received the B.A. degree from the University of Oxford,
the M.Sc. degree from the University of London, and the Ph.D. degree in
queueing theory from the University of Strathclyde.
After holding a post-doctoral position with the University of Cambridge, his research interests centered on wireless. In 1993, he participated in the Telstra trial of Qualcomm CDMA in South Eastern Australia. He then joined the Mobile Research Centre, University of South Australia, Adelaide. He was a Researcher at Bell Labs from January 1997 to June 2013. He is currently a Research Professor at Macquarie University, Sydney, Australia, and also a Consultant to Telstra for the past two years. He has over 25 patents in applications for DSL vectoring, wireless networks, and location and tracking. He has received several awards for his work in DSL vectoring and in wireless scheduling. He has held visiting positions including ones at Brown University (Maths), University of Korea (Engineering), and Vrij University (Maths). For the past three years, he has been a STAR Visiting Scholar with the Maths Department, Technical University of Eindhoven, where his collaborative work includes investigations of both CSMA networks and load balancing in queueing networks. Apart from papers in various aspects of telecommunications, he has been an author on various aspects of probability theory, including random Vandermonde matrices, large deviations theory for occupancy models and more recently random CSMA networks and load balancing in queueing networks. His current research includes storage systems and wireless mmwave networks.
\end{IEEEbiographynophoto}

\end{document}